\documentclass{IEEEtran} 

\usepackage{graphics} 
\usepackage{epsfig} 
\usepackage{amsmath} 
\usepackage{amssymb}  

\usepackage{algpseudocode}
\usepackage{algorithm}
\usepackage{epstopdf}

\usepackage{hyperref}

\usepackage{amsthm}

\newtheorem{assumption}{\bf Assumption}

\newtheorem{theorem}{\bf Theorem}
\newtheorem{proposition}{\bf Proposition}
\newtheorem{lemma}{\bf Lemma}
\newtheorem{corollary}{\bf Corollary}
\newtheorem{remark}{\bf Remark}

\usepackage{setspace}
\usepackage{tikz}               
\usepgflibrary{arrows}

\usepackage{enumitem}

\begin{document}
\IEEEoverridecommandlockouts
\IEEEpubid{\begin{minipage}{\textwidth}\ \\[12pt] \\ \\
         \copyright 2020 IEEE.  Personal use of this material is  permitted.  Permission from IEEE must be obtained for all other uses, in  any current or future media, including reprinting/republishing this material for advertising or promotional purposes, creating new  collective works, for resale or redistribution to servers or lists, or  reuse of any copyrighted component of this work in other works.
     \end{minipage}}
\title{A computationally efficient robust model predictive control framework for uncertain nonlinear systems  \\- extended version}
\author{Johannes K\"ohler$^1$, Raffaele Soloperto$^1$, Matthias A. M\"uller$^2$, Frank Allg\"ower$^1$
\thanks{%
This work was supported by the German Research Foundation under Grants GRK 2198/1 - 277536708, AL 316/12-2, and MU 3929/1-2 - 279734922. 
The authors thank the International Max Planck Research School for Intelligent Systems (IMPRS-IS) for supporting Raffaele Soloperto.
}
\thanks{$^1$Johannes K\"ohler, Raffaele Soloperto,   and Frank Allg\"ower are with the Institute for Systems Theory and Automatic Control, University of Stuttgart, 70550 Stuttgart, Germany.
(email:$\{$ johannes.koehler, raffaele.soloperto, frank.allgower\}@ist.uni-stuttgart.de)}
\thanks{$^2$Matthias A. M\"uller is with the Institute of Automatic Control, Leibniz University Hannover, 30167 Hannover, Germany (e-mail:mueller@irt.uni-hannover.de)}
}

\maketitle
%
\begin{abstract}
In this paper, we present a nonlinear robust model predictive control (MPC) framework for general (state and input dependent) disturbances.
This approach uses an online constructed tube in order to tighten the nominal (state and input) constraints. 
To facilitate an efficient online implementation, the shape of the tube is based on an offline computed incremental Lyapunov function with a corresponding (nonlinear) incrementally stabilizing feedback. 
Crucially, the online optimization only implicitly includes these nonlinear functions in terms of scalar bounds, which enables an efficient implementation.   
Furthermore, to account for an efficient evaluation of the worst case disturbance, a simple function is constructed offline that upper bounds the possible disturbance realizations in a neighbourhood of a given point of the open-loop trajectory. 
The resulting MPC scheme ensures robust constraint satisfaction and practical asymptotic stability with a moderate increase in the online computational demand compared to a nominal MPC. 
We demonstrate the applicability of the proposed framework in comparison to state of the art robust MPC approaches with a nonlinear benchmark example. 
This paper is an extended version of~\cite{koehler2020robust}, and contains further details and additional considers: continuous-time systems (App.~\ref{app:cont}), more general nonlinear constraints  (App.~\ref{app:nonlin_constraints}) and special cases (Sec.~\ref{sec:RMPC_special}). 
\end{abstract}


\section{Introduction}
\subsection*{Motivation}
Model Predictive Control (MPC)~\cite{rawlings2017model,kouvaritakis2016model,grune2017nonlinear} is an optimization based control method that can handle general nonlinear dynamics and constraints. 
There exists a great body of literature on the design of MPC schemes that ensure rigorous theoretical properties (such as recursive feasibility, constraint satisfaction and stability), assuming that an accurate model of the plant is available. 
 Designing MPC schemes that ensure these properties despite disturbances and/or uncertainty in the model is the topic of robust MPC. 

There exist sophisticated methods to account for uncertainty and disturbances using Min-Max MPC~\cite{raimondo2009min}, scenario MPC~\cite{calafiore2013robust,schildbach2014scenario} and stochastic MPC~\cite{lorenzen2017stochastic,hewing2019cautious}. 
These schemes enjoy rigorous theoretical properties and can yield good performance, but typically suffer from a significantly increased online computational demand. 
Under suitable conditions, nominal MPC schemes have inherent robustness properties \cite{yu2014inherent,allan2017inherent}, however, in the presence of hard state constraints this robustness margin can be arbitrary small or even nonexistent~\cite{grimm2004examples}. 

The practical compromise (in terms of performance and complexity) is robust tube-based MPC, which uses a tube around the nominal trajectory that confines the actual (uncertain) system trajectory.
Robust constraint satisfaction is ensured by tightening the nominal constraints based on this tube.
The result is a certainty equivalent MPC, which uses a nominal prediction model in combination with constraint tightening. 
 In this paper we present a tube-based robust MPC scheme that is applicable to a large class of nonlinear systems subject to general state and input dependent uncertainty. 

\IEEEpubidadjcol

\subsection*{Related work}
In~\cite{chisci2001systems,mayne2005robust} tube based robust MPC schemes for linear systems subject to additive disturbances are presented,  based on an auxiliary linear controller in combination with a polytopic tube. 
In order to extend this approach to nonlinear systems, a corresponding auxiliary controller together with a simple method to compute a tube is necessary, typically in form of sublevel sets of some incremental Lyapunov function. 
In~\cite{marruedo2002input,limon2005robust} no stabilizing feedback is considered and a tube is constructed based on a Lipschitz constant of the dynamics or interval arithmetics, respectively.  
Although simple to apply, these approaches become very conservative for larger prediction horizons. 
In~\cite{yu2013tube} a quadratic incremental Lyapunov function with a linear auxiliary controller are computed offline, which are used to design a robust MPC scheme that is applicable to nonlinear systems with small nonlinearities only. 
In~\cite{DynamicTube_Lopez_19} for the special case of feedback linearizable systems, the tube is parametrized as a hyper cube and a boundary layer controller is used as an auxiliary controller.
In~\cite{bayer2013discrete} incrementally stable systems are considered and the sublevel sets of the incremental Lyapunov function are used to characterize the tube and thus tighten the constraints. 
Similarly, in~\cite{singh2017robust} a control contraction metric is determined offline by solving a sum-of-squares problem, which is used to tighten the constraints, compare also \cite{wang2019differential}. 
In~\cite{villanueva2017robust}, a tube is computed online within the MPC optimization problem based on Min-Max differential inequalities. 
In~\cite{mayne2011tube} a reference tracking MPC scheme is used as an ancillary controller to stabilize a pre-computed nominal trajectory with tightened constraints.

Similar problems and methods can also be found in robust trajectory optimization in robotics. 
For example, in~\cite{majumdar2017funnel} complex movements are decomposed into simple maneuvers and for each maneuver a funnel/tube is computed offline.  
In~\cite{manchester2017dirtrel}, the time-varying (iterative) LQR is used to tighten the constraints, which is similar to~\cite{villanueva2017robust}. 

In general, by using more sophisticated concepts, robust MPC methods can handle more general nonlinear systems and reduce the conservatism. 
However, this comes at the cost of an increased offline and/or online computational complexity. 
In this paper, we consider the quite general class of nonlinear systems which are incrementally stabilizable~\cite{koehler2020nonlinear}. Instead of directly using the complex (and possibly unknown) incremental Lyapunov function in the design of the MPC scheme, we only use scalar bounds that characterize the stabilizability property.
This enables us to use a simple and efficient implementation to robustly stabilize complex nonlinear systems.

In addition to the challenges related to nonlinear system dynamics, the consideration of general state and input dependent disturbance descriptions instead of additive constantly bounded disturbances is a crucial problem. 
In the special case of linear systems subject to parametric uncertainty, a polytopic tube can be constructed online~\cite{fleming2015robust,houska2016short},  \cite{feng2019min}, \cite[Chap.~5]{kouvaritakis2016model}. 
For nonlinear systems, only very few results allowing for general state and input dependent uncertainty descriptions are available. In~\cite{pin2009robust} the method in~\cite{marruedo2002input} is extended to state dependent disturbances by constructing a simple tube online. 
Similar to~\cite{marruedo2002input}, the procedure is easy to apply but suffers from a prohibitive conservatism.
As a complementary result, in~\cite{RS_NMPC_18} an offline constraint tightening based on backward reachable sets is proposed.
More advanced methods are based on scenario or stochastic formulations~\cite{calafiore2013robust,schildbach2014scenario,lorenzen2017stochastic,hewing2019cautious}, compare also \cite{lucia2013multi}.  
These approaches typically suffer from a significant increase in the online computational demand and yield probabilistic guarantees, if any.

\subsection*{Contribution}
In this work, we present a nonlinear robust MPC framework for incrementally (exponentially) stabilizable nonlinear systems subject to general nonlinear state and input dependent disturbances/uncertainty. 
The paper contains two main contributions.
First, we present a simple constraint tightening for nonlinear robust MPC based on incremental stabilizability. 
Second, we provide a framework to consider general nonlinear state and input dependent disturbances in robust MPC by including the predicted size of the tube as scalar variables in the online MPC optimization problem. 

Most of the existing design procedures for nonlinear robust MPC~\cite{yu2013tube,bayer2013discrete,singh2017robust} compute a complex nonlinear incremental Lyapunov function and a corresponding incrementally stabilizing feedback offline and then use them to design the online optimization problem. 
In contrast, we only use the fact that the system is incrementally stabilizable, i.e., the existence of a possibly quite complex (or analytically unknown) nonlinear incremental Lyapunov function with a corresponding feedback, and then design the robust MPC based on suitable scalar bounds on these functions, which can be obtained numerically.  
This enables a simple and efficient implementation. 

Furthermore, we design a general nonlinear state and input dependent function that characterizes uncertainty for different operating points. 
Based on this description, we augment the nonlinear robust MPC scheme such that the size of the tube and correspondingly the constraint tightening are computed online (depending on the nominal predicted trajectory). 
The resulting nonlinear robust MPC scheme ensures robust recursive feasibility and robust constraint satisfaction, while avoiding highly uncertain areas, thus acting cautiously.

Within the proposed framework, the offline computation only requires simple scalar functions/operations, and the online computational demand is only moderately increased compared to a nominal MPC scheme, similar to~\cite{marruedo2002input,pin2009robust}. 
We demonstrate  the applicability of the proposed approach and compare it to competing approaches~\cite{marruedo2002input,villanueva2017robust,pin2009robust} in terms of computational complexity and conservatism with a nonlinear example.

The paper is structured as follows: 
Section~\ref{sec:setup_RMPC} discusses the problem setup.  
Section~\ref{sec:RMPC_main} presents the proposed scheme and the theoretical analysis. 
Section~\ref{sec:RMPC_special} discusses the important special cases of constantly bounded additive disturbances and linear parameter varying (LPV) systems.
Section~\ref{sec:num} provides a numerical example to demonstrate the applicability of the proposed approach and compare it to competing approaches~\cite{marruedo2002input,villanueva2017robust,pin2009robust} in terms of computational complexity and conservatism.
Section~\ref{sec:sum} concludes the paper.
In the Appendix, the results are extended to continuous-time dynamics (App.~\ref{app:cont}) and more general nonlinear constraints  (App.~\ref{app:nonlin_constraints}).
A preliminary version for the special case of additive disturbances (Sec.~\ref{sec:additive}) can be found in the conference proceedings~\cite{kohler2018novel}. 

\subsection*{Notation}
The quadratic norm with respect to a positive definite matrix $Q=Q^\top$ is denoted by $\|x\|_Q^2=x^\top Q x$ and 
the minimal and maximal eigenvalue of $Q$ are denoted by $\lambda_{\min}(Q)$ and $\lambda_{\max}(Q)$, respectively. 
The positive real numbers are $\mathbb{R}_{\geq 0} = \{ r\in\mathbb{R}|r\geq  0\}$.
The vertices of a polytopic set $\Theta$ are denoted by $\theta^i\in\text{Vert}(\Theta)$. 
By $\mathcal{K}$ we denote the class of functions $\alpha:\mathbb{R}_{\geq 0}\rightarrow\mathbb{R}_{\geq 0}$, which are continuous, strictly increasing and satisfy $\alpha(0)=0$. 
By $\mathcal{K}_{\infty}$ we denote the class of functions $\alpha\in\mathcal{K}$, which are also unbounded. 

\section{Problem setup: Incremental stabilizability and uncertainty description}
\label{sec:setup_RMPC}
This section introduces the assumptions on the nonlinear system dynamics and the uncertainty. 
The problem setup is introduced in Section~\ref{sec:setup}. 
 Section~\ref{sec:increm} discusses the system property \textit{incremental stabilizability}, which is key for the proposed simple implementation. 
An efficient description of the disturbance bound is introduced in Section~\ref{sec:dist_description}. 
%
%
\subsection{Setup}
\label{sec:setup}
We consider a nonlinear perturbed discrete-time system 
\begin{align}
\label{eq:sys_w}
x_{t+1}&=f_w(x_t,u_t,d_t)=f(x_t,u_t)+d_w(x_t,u_t,d_t),
\end{align}
 with state $x\in\mathbb{R}^n$, control input $u\in\mathbb{R}^m$, disturbance $d\in\mathbb{D}\subset\mathbb{R}^q$, time $t\in\mathbb{N}$, perturbed system $f_w$, nominal model $f$ and model mismatch $d_w$.
We impose point-wise in time state and input constraints  
\begin{align}
\label{eq:constraint}
(x_t,u_t)\in \mathcal{Z},\quad t\geq 0,
\end{align}
with some compact nonlinear constraint set 
\begin{align*}
\mathcal{Z}=\{(x,u)\in\mathbb{R}^{n+m}|~g_j(x,u)\leq 0,~j=1,\dots,p\}\subset\mathbb{R}^{n+m}.
\end{align*}
\begin{assumption}
\label{ass:hat_w}
For each $(x,u)\in\mathcal{Z}$, there exists a compact set $\mathcal{W}(x,u)\subset\mathbb{R}^n$, such that the model mismatch $d_w$ satisfies $d_w(x,u,d)\in\mathcal{W}(x,u)$ for all $d\in\mathbb{D}$. 
\end{assumption}
Such a description includes additive disturbances, multiplicative disturbances, more general nonlinear disturbances, and/or unmodeled nonlinearities. 
We consider the problem of stabilizing the origin and assume that $0\in\text{int}(\mathcal{Z})$, $f(0,0)=0$. 
The open-loop cost of a predicted state and input sequence $x_{\cdot|t}\in \mathbb{R}^{n\cdot(N+1)}$, $u_{\cdot|t}\in\mathbb{R}^{m\cdot N}$ is defined as
\begin{align*}
&J_N(x_{\cdot | t},u_{\cdot|t})=\sum_{k=0}^{N-1}\ell(x_{k|t},u_{k|t})+V_f(x_{N|t}),
\end{align*}
with some positive definite stage cost $\ell$ and terminal cost $V_f$. 
One conceptual framework to address this robust constrained stabilization problem is $\min-\max$ MPC~\cite{raimondo2009min}:    
\begin{subequations}
\label{eq:MPC_concept}
\begin{align}
&\min_{u_{\cdot|t}}\max_ {w_{\cdot|t}} J_N(x_{\cdot |t},u_{\cdot|t})\\
\text{s.t. }
&x_{0|t} = x_t,\quad x_{k+1|t} = f(x_{k|t}, u_{k|t})+w_{k|t},\\
&(x_{k|t},u_{k|t})\in\mathcal{Z},~
 x_{N|t} \in \mathcal{X}_f,\\
& \forall w_{k|t}\in\mathcal{W}(x_{k|t},u_{k|t}),\quad \nonumber
 k=0,\dots, N-1.\nonumber
\end{align}
\end{subequations}
The result of this optimization problem is a control input\footnote{
Typically, $\min-\max$ MPC schemes~\cite{raimondo2009min} optimize over feedback policies $u_{k|t}=\Pi_k(x_{k|t})$ instead of open-loop control inputs $u_{k|t}$.
} $u^*_{\cdot|t}$, such that  the constraints are satisfied for the worst-case disturbance realization $w^*_{\cdot|t}$.  
This approach is conceptually very appealing but suffers from a prohibitive computational demand. 
In this paper, we present a framework for tube-based nonlinear robust MPC that shares many of the theoretical properties with~\cite{raimondo2009min}, with an online computational demand comparable to a nominal MPC scheme. 
 
%
\subsection{Local Incremental stabilizability}
\label{sec:increm}
In order to provide theoretical guarantees for robust stabilization, we assume that the nominal system is locally incrementally stabilizable. 
\begin{assumption}
\label{ass:contract}  \cite{koehler2020nonlinear}, \cite[Ass.~1]{kohlernonlinear19}
There exist a control law $\kappa:\mathbb{R}^n\times\mathcal{Z}\rightarrow\mathbb{R}^m$, an incremental Lyapunov function $V_{\delta}:\mathbb{R}^n\times \mathcal{Z}\rightarrow\mathbb{R}_{\geq 0}$, which is continuous in the first argument and satisfies $V_{\delta}(z,z,v)=0$ for all $(z,v)\in\mathcal{Z}$, and positive constants $c_{\delta,l},~c_{\delta,u},~\delta_{\text{loc}},~\kappa_{\max}{>0}$, $\rho\in(0,1)$, such that the following properties hold for all $(x,z,v)\in\mathbb{R}^n\times\mathcal{Z}$ with $V_{\delta}(x,z,v)\leq \delta_{\text{loc}}$, and all $(x^+,z^+,v^+)\in\mathbb{R}^n\times\mathcal{Z}$: 
\begin{subequations}
\label{eq:contract_ass}
\begin{align}
\label{eq:bound}
c_{\delta,l}\|x-z\|^2\leq V_{\delta}(x,z,v)\leq& c_{\delta,u}\|x-z\|^2,\\
\label{eq:k_max}
\|\kappa(x,z,v)-v\|^2\leq& \kappa_{\max}{{V_{\delta}(x,z,v)}},\\
\label{eq:contract}
V_{\delta}(x^+,z^+,v^+)\leq& \rho^2 V_{\delta}(x,z,v),
\end{align}
\end{subequations}
 with $x^+=f(x,\kappa(x,z,v))$, $z^+=f(z,v)$.  
\end{assumption}
\begin{remark}
\label{rk:stab}
It is possible to relax Assumption~\ref{ass:contract}, such that for any feasible sequence $z_{t+1}=f(z_t,v_t)$, $(z_t,v_t)\in\mathcal{Z}$, $t\geq 0$, there exist a time-varying incremental Lyapunov function $V_{\delta,t}(x,z)$ and a time-varying feedback $\kappa_t(x,z,v)$, that (uniformly) satisfy the conditions~\eqref{eq:contract_ass}. 
In this paper, we consider the time-invariant description in Assumption~\ref{ass:contract} for notational simplicity.  
This system property is a natural extension of previous works on incremental stability and corresponding incremental Lyapunov
functions~\cite{koehler2020nonlinear}, \cite[Ass.~1]{kohlernonlinear19}, \cite{angeli2002lyapunov} and contains various stabilizability properties considered in existing robust MPC schemes as special cases.  
In~\cite{villanueva2017robust}  $V_{\delta,t}(x,z)=\|x-z\|_{P_t}^2$, $\kappa_t(x,z,v)=v+K_t(x-z)$ is considered, with $K_t,~P_t$ computed online, similar to the time-varying (iterative) LQR in~\cite{manchester2017dirtrel}.
In~\cite{singh2017robust} $V_{\delta}(x,z)$ and $\kappa(x,z,v)=v+K_z(x-z)$ are determined offline using control contraction metrics~\cite{manchester2017control}. 
Similarly, in~\cite{koehler2020nonlinear} $V_{\delta}(x,z)=\|x-z\|_{P_z}^2$ and $\kappa(x,z,v)=v+K_z(x-z)$ are computed offline using quasi-LPV methods. 
The special case of constant matrices $P,K$ is considered in~\cite{yu2013tube,yu2010robust}. 
In~\cite{bayer2013discrete} the system is assumed to be incrementally stable with $\kappa(x,z,v)=v$.
\end{remark}    
The following assumptions capture the considered conditions on the stage cost $\ell$ and the constraint set $\mathcal{Z}$.
\begin{assumption}
\label{ass:stage_cost}
The stage cost $\ell:\mathcal{Z}\rightarrow\mathbb{R}_{\geq0}$ satisfies
\begin{subequations}
\begin{align}
\label{eq:stage_cost_l}
\ell(r) &\geq \alpha_{\ell}(\|r\|),\\
\label{eq:stage_cost_c}
\ell(\tilde{r})-\ell(r) &\leq \alpha_c(\|\tilde{r}-r\|), \quad \forall~r \in \mathcal{Z},  \tilde{r}\in\mathbb{R}^{n+m},
\end{align}
\end{subequations}
with $\alpha_{\ell},~\alpha_c\in\mathcal{K}_{\infty}$. 
Furthermore, for any $\rho\in(0,1)$, we have\footnote{
For polynomials $\alpha_c(c)=\sum_{j=1}^\infty a_j c^j$, $a_j\geq  0$, with $a_k>0$ for some $k\in\mathbb{N}$, this condition is satisfied with 
$\alpha_{c,\rho}(c):=\sum_{k=0}^{\infty}\sum_{j=1}^\infty a_j (c\rho^k)^j=\sum_{j=1}^\infty {a_j c^j}/({1-\rho^j})$.
} $\alpha_{c,\rho}(c):=\sum_{k=0}^{\infty} \alpha_c(\rho^k c)\in\mathcal{K}_{\infty}$. 
\end{assumption}
\begin{assumption}
\label{ass:gen_nonlin_con_Lipschitz} 
There exist local Lipschitz constants $L_j$, such that 
\begin{align}
\label{eq:nonlin_con_Lipschitz}
g_j(\tilde{r})-g_j(r)\leq L_j\|r-\tilde{r}\|, \quad j=1,\dots,p, 
\end{align}
holds for all $r\in\mathcal{Z}$ and all $\tilde{r}\in\mathbb{R}^{n+m}$ with $\|r-\tilde{r}\|^2\leq \frac{\delta_{loc}}{c_{\delta,l}}$.
\end{assumption}
Assumptions~\ref{ass:stage_cost}--\ref{ass:gen_nonlin_con_Lipschitz} are, for example, satisfied with a quadratic positive definite stage cost $\ell$ and a convex polytopic constraint set $\mathcal{Z}$. 
More general stage costs are discussed in Remark~\ref{remark:general_cost} (Sec.~\ref{sec:discuss}). 
In Appendix~\ref{app:nonlin_constraints}, we discuss the extension to more general constraints, which do not satisfy Assumption~\ref{ass:gen_nonlin_con_Lipschitz}.

The following proposition allows us to compute scalar bounds that relate the level set of the incremental Lyapunov function $V_{\delta}$ to the nonlinear constraint set $\mathcal{Z}$. 
\begin{proposition}
\label{prop:contract}
Suppose that Assumptions \ref{ass:contract}, \ref{ass:stage_cost} and \ref{ass:gen_nonlin_con_Lipschitz} hold, then there exist constants $c_j\geq 0$, $j=1,\dots, p$, and a function $\alpha_u \in \mathcal{K}_{\infty}$ such that the following inequalities hold for all $(x,z,v)\in\mathbb{R}^n\times\mathcal{Z}$ with $V_{\delta}(x,z,v)\leq c^2$ and any $c\in[0,\sqrt{\delta_{\text{loc}}}]$:
\begin{align}
\label{eq:stage_cost_bound}
\ell(x, \kappa(x,z,v))-\ell(z,v) \leq&  \alpha_u(c),\\ 
\label{eq:lipschitz}
g_j(x,\kappa(x,z,v))-g_j(z,v)\leq& c_j\cdot c. 
\end{align}
\end{proposition}
\begin{proof}
Based on \eqref{eq:bound} and \eqref{eq:k_max}, we have
\begin{align}
\label{eq:triangle_in}
\|x-z\|^2+\|\kappa(x,z,v)-v\|^2 \leq \left(\frac{1}{{c_{\delta, l}}}+ \kappa_{\max}\right) c^2. 
\end{align}
This implies \eqref{eq:stage_cost_bound} with
\begin{align*}
&\ell(x, \kappa(x,z,v))-\ell(z,v) \\
\stackrel{\eqref{eq:stage_cost_c}} \leq& \alpha_c\left(c\sqrt{ \left({1}/{{c_{\delta, l}}} + \kappa_{\max}\right) }\right) 
 =:  \alpha_u(c).
\end{align*}
Similarly, \eqref{eq:lipschitz} is satisfied with
\begin{align*}
&g_j(x,\kappa(x,z,v))-g_j(z,v) \\
\stackrel{\eqref{eq:nonlin_con_Lipschitz}}{\leq}& L_j\sqrt{\|x-z\|^2+\|\kappa(x,z,v)-v\|^2} \\
\stackrel{\eqref{eq:triangle_in}}{\leq} &L_j c\cdot \sqrt{ 1/c_{\delta, l} + \kappa_{\max}}=:c_j \cdot c.  
\end{align*}
\end{proof}
One of the most important features of the proposed robust MPC scheme is that the explicit description and/or computation of the general nonlinear functions $V_{\delta},~\kappa$ is not required for the implementation. 
Thus, this scheme can also be applied if an analytic description and/or verification of this system property are too complex. 
Indeed, only the scalar variables $c_j,~c_{\delta,l},~c_{\delta,u},~\delta_{\text{loc}}$, and $\rho$ need to be computed. 
For example, control contraction metrics (CCM)~\cite{manchester2017control} provide simple (analytical) differential Lyapunov functions and differential feedbacks, which allow for a simple computation of suitable constants (c.f.~\cite[Thm.~III.2]{singh2017robust}). However, evaluating the corresponding control law $\kappa$ and incremental Lyapunov function $V_{\delta}$ is, in general, complex and requires the solution of a nonlinear optimization problem to determine the minimizing geodesic. 
The simple description requiring only knowledge of the above specified constants ensures that the online computational demand does not increase significantly, which is one of the major advantages of the proposed approach. 
We are currently investigating numerical approaches to directly compute these scalar values without using an explicit analytical description of $V_{\delta}$. 
In addition, these scalar variables can be easily tuned to provide safety, compare the numerical example in~\cite{wabersich2018safe}.

\subsection{Efficient disturbance description}
\label{sec:dist_description}
One of the difficulties of considering general state and input dependent disturbances is that the future size of the disturbance depends on the predicted state and input sequence $(x_{\cdot|t},u_{\cdot|t})$. 

A safe and reliable (robust) MPC implementation requires that a valid upper bound on the magnitude of the disturbance at the predicted (uncertain) state $x_{k|t}$ 
can be efficiently evaluated.
To facilitate such an efficient evaluation we consider the following assumption.
\begin{assumption}
\label{ass:approx}
Consider the uncertainty set $\mathcal{W}(x,u)$, the incrementally stabilizing feedback $\kappa$ and incremental Lyapunov function $V_{\delta}$ from Assumptions~\ref{ass:hat_w} and~\ref{ass:contract}. 
There exists a function $\tilde{w}_{\delta}:\mathcal{Z}\times\mathbb{R}_{\geq 0}\rightarrow\mathbb{R}_{\geq0}$, such that for any point $(x,z,v)\in\mathbb{R}^n\times\mathcal{Z}$ with $V_{\delta}(x,z,v)\leq c^2$, any $c\in[0,\sqrt{\delta_{\text{loc}}}]$, any $(z^+,v^+)\in\mathcal{Z}$ with $z^+=f(z,v)$, and any disturbance $d_w\in\mathcal{W}(x,\kappa(x,z,v))$, we have
\begin{subequations}
\begin{align}
\label{eq:w_delta_tilde_1}
{V_{\delta}(z^++d_w,z^+,v^+)}\leq \tilde{w}^2_{\delta}(z,v,c).
\end{align}
Furthermore, $\tilde{w}_{\delta}$ satisfies the following monotonicity property: For any point $(x,z,v)\in\mathbb{R}^n\times \mathcal{Z}$ such that $V_{\delta}(x,z,v)\leq (c_1-c_2)^2$ with constants $0\leq c_2\leq c_1\leq \sqrt{\delta_{loc}}$, we have 
\begin{align}
\label{eq:w_delta_tilde}
\tilde{w}_{\delta}(x,\kappa(x,z,v),c_2)\leq \tilde{w}_{\delta}(z,v,c_1).
\end{align}
\end{subequations}
\end{assumption}
Note that by using the set $\mathcal{W}(x,u)$ in Assumption~\ref{ass:hat_w}, we can construct a function that upper bounds the uncertainty at any point $(x,u)\in\mathcal{Z}$. 
In a similar fashion,~\eqref{eq:w_delta_tilde_1} ensures that $\tilde{w}_{\delta}(z,v,c)$ is an upper bound to the uncertainty that can occur in the neighbourhood of a point $(z,v)\in\mathcal{Z}$, where the neighbourhood is given by $V_{\delta}(x,z,v)\leq c^2$.  
The monotonicity property~\eqref{eq:w_delta_tilde} is illustrated in Fig.~\ref{fig:monotonicity}, where we have that the uncertainty bound based on the set $\mathcal{S}_1$ needs to be larger than the uncertainty contained in the set $\mathcal{S}_2$. 
This property is crucial to evaluate bounds on the uncertainty for the online optimization, which remain valid in closed-loop operation. 
We refer to this as a \textit{monotonicity} property since in the quadratic case $\mathcal{S}_2 \subseteq \mathcal{S}_1$, and thus the measure-like function $\tilde{w}_\delta$ should satisfy \eqref{eq:w_delta_tilde}. 
Details on the offline computation/construction of the function $\tilde{w}_{\delta}$ can be found in Section~\ref{sec:description_details} (Prop.~\ref{prop:w_tilde}--\ref{prop:w_tilde_2}, Corollary~\ref{corol:w_tilde}),  the numerical example (Sec.~\ref{sec:num}) and in Section~\ref{sec:param} (Prop.~\ref{prop:LPV_simple}) for the special case of LPV systems. 
\begin{figure}[hbtp]
\begin{center}
\includegraphics[scale=0.48]{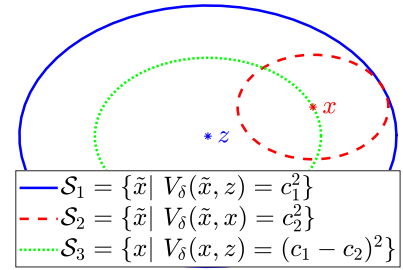}
\end{center}
\caption{Illustration: Monotonicity property of $\tilde{w}_{\delta}$ for quadratic functions $V_{\delta}$.}
\label{fig:monotonicity}
\end{figure}

\section{Robust MPC framework}  
\label{sec:RMPC_main}
This section contains the proposed robust MPC framework for nonlinear uncertain systems. 
The proposed framework is presented in Section~\ref{sec:RMPC_scheme}. 
The theoretical analysis, including assumptions on the terminal ingredients, is detailed in Section~\ref{sec:theory}. 
Section~\ref{sec:discuss} discusses the complexity and conservatism compared to existing robust MPC methods. 
Additional details on the offline computation are given in Section~\ref{sec:description_details} and the overall offline and online procedure is summarized. 
%
\subsection{Proposed nonlinear robust MPC scheme}
\label{sec:RMPC_scheme}
In the following, we present the proposed robust MPC scheme.  
The basic idea of the proposed approach is to online predict a tube size $s\in\mathbb{R}_{\geq 0}$, which (indirectly) characterizes sublevel sets of the incremental Lyapunov function $V_{\delta}$ (Ass.~\ref{ass:contract}). 
The tube size $s$ is then used to tighten the state and input constraints in order to ensure robust constraint satisfaction. 
The proposed nonlinear robust MPC scheme is based on the following optimization problem: 
\begin{subequations}
\label{eq:MPC_real}
\begin{align}
V_N(x_t)=&\min_{u_{\cdot|t},w_{\cdot|t},x_{\cdot|t},s_{\cdot|t}}J_N(x_{\cdot |t},u_{\cdot|t})\\
\text{s.t. }
&x_{0|t} = x_t,\quad 
\label{eq:tube_definition_init}
s_{0|t} = 0,\\
&x_{k+1|t} = f(x_{k|t}, u_{k|t}),\\
\label{eq:tube_definition}
&s_{k+1|t}=\rho s_{k|t}+{w}_{k|t},\\
\label{eq:w_definition}
&{w}_{k|t}\geq \tilde{w}_{\delta}(x_{k|t},u_{k|t},s_{k|t}),\\
\label{eq:real_constraints}
&g_j(x_{k|t},u_{k|t})+c_js_{k|t}\leq 0,\\
\label{eq:tube_bound_MPC}
&s_{k|t}\leq \overline{s},\quad 
 {w}_{k|t}\leq \overline{w},\\
\label{eq:terminal_region}
& (x_{N|t},s_{N|t}) \in \mathcal{X}_f,\\
& k=0,\dots, N-1,\quad j=1,\dots, p.\nonumber
\end{align}
\end{subequations}
The solution of \eqref{eq:MPC_real} are the optimal trajectories for the state $x^*_{\cdot|t}$, the input $u^*_{\cdot|t}$, the disturbance bound ${w}^*_{\cdot|
t}$, the tube size ${s}^*_{\cdot|t}$ and the value function $V_N$.
The resulting closed-loop system is given by
\begin{align}
\label{eq:close}
x_{t+1}=f_w(x_t,u_t,d_t),\quad u_t=u^*_{0|t}.
\end{align} 
The terminal cost $V_f$, the terminal set $\mathcal{X}_f$ and the scalar bounds $\overline{s},~\overline{w}$ will be introduced in Section~\ref{sec:theory}.  
Compared to a nominal MPC scheme, the state $x$ is augmented with the scalar variable $s$, which captures the size of the tube, while the input $u$ is augmented with the scalar variable $w$, which captures the magnitude of the uncertainty. 
Similarly to the state $x$ and input $u$, the additional state $s$ and input $w$ are subject to nonlinear dynamic equations~\eqref{eq:tube_definition}, \eqref{eq:w_definition}, constraints~\eqref{eq:real_constraints}, \eqref{eq:tube_bound_MPC} and terminal condition~\eqref{eq:terminal_region}. 
Correspondingly, the online computational demand of solving~\eqref{eq:MPC_real} is equivalent to a nominal MPC scheme with an augmented state $(x,s)\in\mathbb{R}^{n+1}$, augmented input vector $(u,w)\in\mathbb{R}^{m+1}$ and additional nonlinear inequality constraints~\eqref{eq:w_definition} on the decision variable $w$.

The numerical example (Sec.~\ref{sec:num}) also shows that the proposed robustification (based on the online computed tube size $s$ and the constraint tightening) only moderately increases the online computational demand.  
This is in contrast to much of the existing approaches for the considered general setup, such as~\cite{calafiore2013robust}, \cite{schildbach2014scenario,hewing2019cautious}, \cite{villanueva2017robust}, \cite{lucia2013multi} which often increase the online computational demand by orders of magnitude.

The scalar $\overline{w}$ represents a sufficient bound on the uncertainty that can occur in closed-loop operation, such that we can guarantee robust recursive feasibility of the terminal set constraint.  
Note that we do not assume any upper bound on the model mismatch $d_w$ in $\mathcal{Z}$, instead, the closed-loop system will automatically avoid regions with uncertainty exceeding $\overline{w}$ and, thus, act \textit{cautiously}.
Thus, even in the absence of state/input constraints, the robust MPC scheme avoids regions with large uncertainty to ensure robust closed-loop properties, which can be very advantageous from a practical point of view. 
Typically, such a behaviour is only present in scenario/stochastic MPC~\cite{calafiore2013robust,schildbach2014scenario,lorenzen2017stochastic,hewing2019cautious} by minimizing the expected stage cost $\ell$, compare Remark~\ref{remark:general_cost}. 
Note that due to the general nonlinear dependence of the uncertainty $w$ on the state and input, the tube size $s_{k|t}$ is not necessarily monotonically increasing over $k$, meaning that we could potentially have $s_{k+1|t}<s_{k|t}$.

\subsection{Theoretical analysis} 
\label{sec:theory}
\begin{figure}[tbp]
\begin{center}
\includegraphics[width=0.48\textwidth]{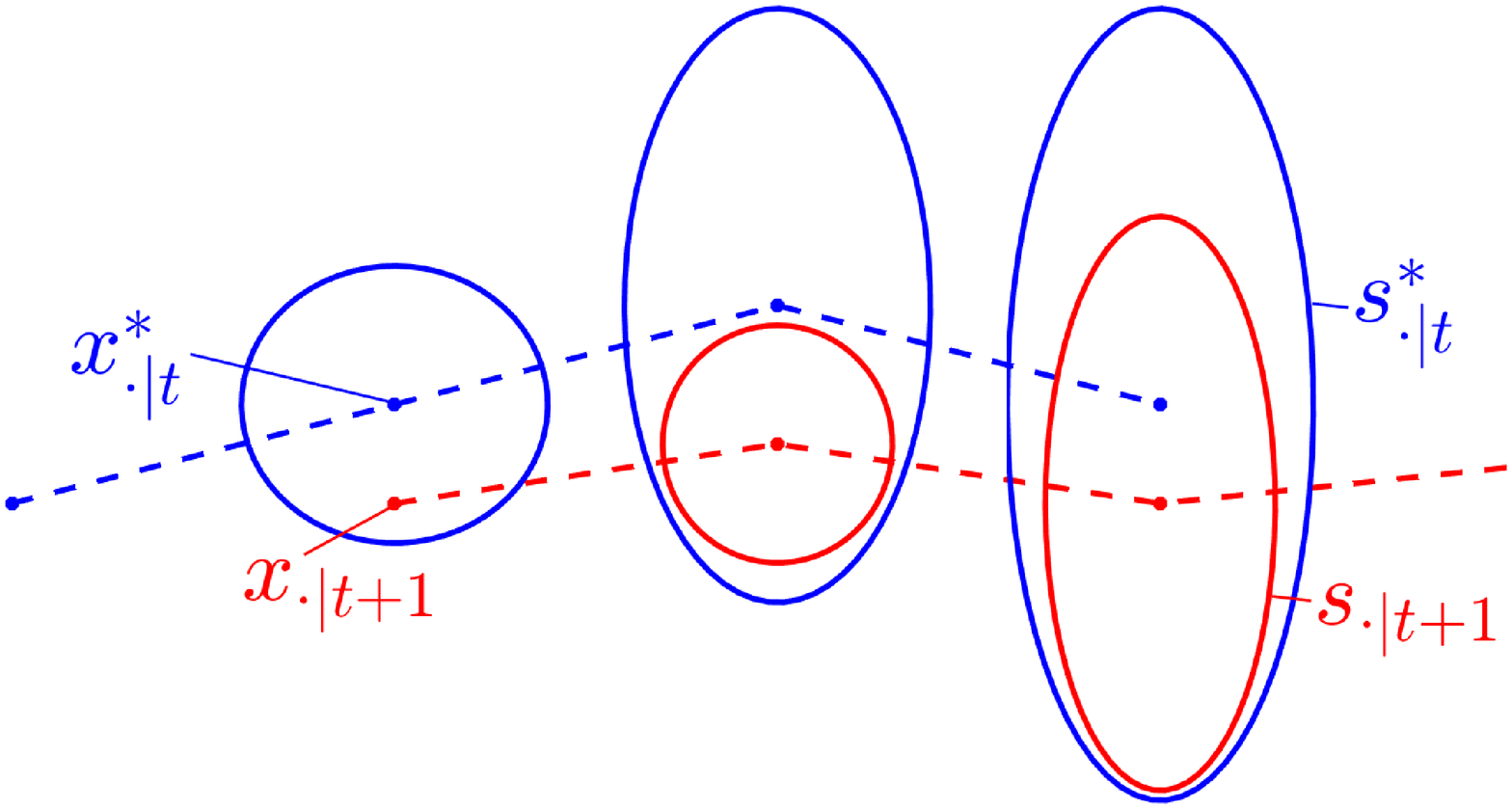}
\end{center}
\caption{Illustration: Optimal trajectory $x^*_{\cdot|t}$, candidate trajectory $x_{\cdot|t+1}$, and corresponding tube size $s^*_{\cdot|t}$ and $s_{\cdot|t+1}$, respectively.}
\label{fig:candidate}
\end{figure}
In the following, we detail the theoretical analysis and provide the technical conditions on the terminal ingredients (terminal cost $V_f$ and terminal set $\mathcal{X}_f$).  
The minimal bound on the uncertainty $\overline{w}_{\min}$ and the maximal\footnote{%
The bound $\overline{s}\leq \sqrt{\delta_{loc}}$ can be relaxed by introducing two constants $\delta_1,\delta_2>0$, such that the contraction~\eqref{eq:contract} holds for all $V_{\delta}(x,z,v)\leq \delta_1^2$, while Assumption~\ref{ass:approx} holds for all $c_2\leq c_1\leq \delta_2$, $c_1-c_2\leq \delta_1$. In this case it suffices that $\overline{s}\leq \delta_2$ and $\overline{w}\leq \delta_1$. 
} tube size $\overline{s}$ are 
\begin{align}
\label{eq:w_min_bound_bar_s}
\overline{w}_{\min}&: = \inf_{(x, u) \in \mathcal{Z}} \tilde{w}_{\delta}(x,u,0),\quad 
\overline{s} :=\sqrt{\delta_{loc}}.  
\end{align}
\subsubsection*{Terminal ingredients}
\label{sec:term}
The following assumption captures the desired properties of the terminal ingredients. 
\begin{assumption}
\label{ass:term_2} 
There exist a terminal controller $k_f:\mathbb{R}^n\rightarrow\mathbb{R}^m$, a terminal cost function $V_f:\mathbb{R}^n\rightarrow \mathbb{R}_{\geq 0}$, a terminal set $\mathcal{X}_f\subset\mathbb{R}^{n+1}$, and a constant $\overline{w} \in \mathbb{R}_{\geq 0}$ such that the following properties hold for all $(x,s)\in\mathcal{X}_f$, all $w\in[\overline{w}_{\min},\overline{w}]$, all $s^+\in[0, \rho s-\rho^N w+\tilde{w}_{\delta}(x,k_f(x), s)]$, and all $d_w\in\mathbb{R}^n$, such that $V_{\delta}(x^+ + d_w,x^+,k_f(x^+))\leq \rho^{2N}w^2$ with $x^+=f(x,k_f(x))$:
\begin{subequations}
\label{eq:term_dist2}
\begin{align}
\label{eq:term_dist_dec2}
V_f(x^+)\leq& V_f(x)-\ell(x,k_f(x)),\\
\label{eq:term_dist_RPI2}
(x^+ + d_w,s^+)\in&\mathcal{X}_f, \\
\label{eq:w_f}
\tilde{w}_{\delta}(x,k_f(x), s)\leq & \overline{w},\\ 
\label{eq:term_dist_con2}
g_j(x,k_f(x))+c_j {s}\leq & 0,\quad j=1,\dots p,\\
\label{eq:term_s}
s\leq& \overline{s}.
\end{align}
\end{subequations}
Furthermore, the terminal cost $V_f$ is continuous on the compact set $\mathcal{X}_{f,x}:=\{x|~\exists s\in[0,\overline{s}],~(x,s)\in\mathcal{X}_f\}$, i.e., there exists a function $\alpha_f\in\mathcal{K}_{\infty}$ such that  
\begin{align}
\label{eq:term_contin2}
V_f(z)\leq V_f(x)+\alpha_f(\|x-z\|),~\forall x,z\in\mathcal{X}_{f,x}.
\end{align}
\end{assumption}
This assumption requires a suitable terminal cost $V_f$, a controller $k_f$, a set $\mathcal{X}_f$ and a scalar $\overline{w}$. 
The conditions on $V_f,~k_f,~\mathcal{X}_f$ are similar to standard conditions in (robust) MPC with a terminal set for the augmented state $(x,s)\in\mathbb{R}^{n+1}$, compare~\cite{chen1998quasi}.  
The terminal set needs to satisfy the tightened state and input constraints~\eqref{eq:term_dist_con2}. 
In addition, a robust positive invariance condition of the terminal set needs to be verified~\eqref{eq:term_dist_RPI2}. 
Details on the constructive satisfaction of Assumption~\ref{ass:term_2} based on standard terminal ingredients can be found in Section~\ref{sec:description_details},  Prop.~\ref{prop:term}.

The following theorem establishes the closed-loop properties of the proposed nonlinear robust MPC scheme. 
\begin{theorem}
\label{thm:main}
Let Assumptions~\ref{ass:hat_w}--\ref{ass:term_2} hold, and suppose that~\eqref{eq:MPC_real} is feasible at $t=0$. 
Then~\eqref{eq:MPC_real} is recursively feasible, the constraints~\eqref{eq:constraint} are satisfied and the origin is practically asymptotically stable for the resulting closed-loop system~\eqref{eq:close}.
\end{theorem}
\begin{proof}
The basic idea is to use the control law $\kappa$ from Assumption~\ref{ass:contract} to stabilize the previous optimal solution of~\eqref{eq:MPC_real} and thus bound the cost increase and ensure robust recursive feasibility. This is illustrated in Figure~\ref{fig:candidate}, where we see that the tube around the candidate solution $x_{\cdot|t+1}$ is contained inside the tube around the previous optimal solution $x^*_{\cdot | t}$.
We first construct the candidate solution and establish bounds on the size of the tube $s$ and the disturbance $w$.
Then we show that the candidate solution satisfies the tightened state and input constraints~\eqref{eq:real_constraints} and the posed constraints on the disturbance bound $w$ and the tube size $s$~\eqref{eq:tube_bound_MPC}.  
Constraint satisfaction, i.e. $(x_t,u_t)\in\mathcal{Z}$  $\forall t\geq 0$, follows directly from~\eqref{eq:real_constraints} with $k=0$.  
Then, we establish recursive feasibility of the posed terminal set constraint $\mathcal{X}_f$~\eqref{eq:terminal_region}.  
Finally, we establish practical asymptotic stability.\\
\textbf{Part I.} Candidate solution: 
For convenience, define
\begin{align*}
u^*_{N|t}=&k_f(x^*_{N|t}),~u^*_{N+1|t}=k_f(x^*_{N+1|t}),\\
x^*_{N+1|t}=&f(x^*_{N|t},u^*_{N|t}),~
w^*_{N|t} = \tilde{w}_\delta(x^*_{N|t}, u^*_{N|t}, s^*_{N|t}).  
\end{align*}
Consider the candidate solution
\begin{align}
\label{eq:robust_candidate}
 x_{0|t+1}=&x_{t+1}=f_w(x^*_{0|t}, u^*_{0|t}, d_t),\\
u_{k|t+1}=&\kappa(x_{k|t+1},x^*_{k+1|t},u^*_{k+1|t}),\nonumber\\
x_{k+1|t+1}=&f(x_{k|t+1},u_{k|t+1}),\nonumber\\
s_{0|t+1}=&0,\nonumber\\
s_{k+1|t+1}=&\rho s_{k|t+1}+{w}_{k|t+1},\nonumber\\
{w}_{k|t+1}=&\tilde{w}_{\delta}(x_{k|t+1},u_{k|t+1},s_{k|t+1}),\nonumber
\end{align}
with $k=0,\dots,N-1$. 
Assumption~\ref{ass:approx},  \eqref{eq:w_delta_tilde_1} with $c=0$ yields 
\begin{align*}
&\sqrt{V_{\delta}(x_{t+1},x^*_{1|t},u^*_{1|t})}\\
\stackrel{\eqref{eq:w_delta_tilde_1}}{\leq}& \tilde{w}_{\delta}(x^*_{0|t},u^*_{0|t},0)\leq w^*_{0|t}=s^*_{1|t}\stackrel{\eqref{eq:tube_bound_MPC}}{\leq} \overline{s}\stackrel{\eqref{eq:w_min_bound_bar_s}}{\leq}\sqrt{\delta}_{loc}. 
\end{align*}
Using the contractivity~\eqref{eq:contract} (Ass.~\ref{ass:contract}) recursively, we get  
\begin{align}
\label{eq:prop_robust_1_2}
&V_{\delta}(x_{k|t+1},x^*_{k+1|t},u^*_{k+1|t})\\
&\leq \rho^{2k} [{w}^*_{0|t}]^2\leq \delta_{loc},\quad k=0,\dots,N.\nonumber 
\end{align}
\textbf{Part II.} Tube dynamics: In the following we show that the following inequalities hold for $k=0,\dots,N-1$ by induction:
\begin{align}
\label{eq:tube_shrinking}
s_{k|t+1}\leq & {s}^*_{k+1|t}-\rho^{k}{w}^*_{0|t},\\
\label{eq:disturbance_shrinking}
{w}_{k|t+1}\leq & {w}^*_{k+1|t}.
\end{align}
We first show that \eqref{eq:tube_shrinking} and \eqref{eq:disturbance_shrinking} hold for $k=0$, and then we proceed by showing that they also hold for $k+1$. 
\begin{enumerate}
\item Induction start: $k=0$: 
Inequality~\eqref{eq:tube_shrinking} is satisfied with 
\begin{align*}
s_{0|t+1}\stackrel{\eqref{eq:tube_definition_init}}{=}0
\stackrel{\eqref{eq:tube_definition}}{=}&s^*_{1|t}- w^*_{0|t}.
\end{align*}
Consider $c_2 = s_{0|t+1} = 0\leq c_1 = {s}^*_{1|t} = {w}^*_{0|t}\leq \overline{s}\leq \sqrt{\delta_{loc}}$,  and
\begin{align*}
&\sqrt{V_{\delta}(x_{0|t+1},x^*_{1|t},u^*_{1|t})} \leq   {w}^*_{0|t}  = c_1-c_2.
\end{align*}
Thus, using Assumption~\ref{ass:approx}, \eqref{eq:w_delta_tilde} implies that the disturbance ${w}$ satisfies  
\begin{align*}
{w}_{0|t+1}= & \tilde{w}_{\delta}(x_{0|t+1}, u_{0|t+1}, \underbrace{s_{0|t+1}}_{=0})\\
\stackrel{\eqref{eq:w_delta_tilde}}{\leq}  & \tilde{w}_{\delta}(x^*_{1|t}, u^*_{1|t}, {s^*_{1|t}})\stackrel{\eqref{eq:w_definition}}{\leq} {w}^*_{1|t}.
\end{align*}
\item Induction step $k+1$: 
Suppose that \eqref{eq:tube_shrinking} and \eqref{eq:disturbance_shrinking} hold for some $k\geq 0$. 
Based on~\eqref{eq:tube_definition}, the tube size $s$ satisfies
\begin{align*}
&s_{k+1|t+1}\stackrel{\eqref{eq:tube_definition}}{=}\rho s_{k|t+1}+{w}_{k|t+1}
\stackrel{\eqref{eq:disturbance_shrinking}}{\leq}\rho s_{k|t+1}+{w}^*_{k+1|t}\\
\stackrel{\eqref{eq:tube_shrinking}}{\leq}& \rho s^*_{k+1|t}+w^*_{k+1|t}-\rho^{k+1}w^*_{0|t}
\stackrel{\eqref{eq:tube_definition}}{=}s^*_{k+2|t}-\rho^{k+1}w^*_{0|t}. 
\end{align*}
Consider $c_2 = s_{k+1|t+1} \stackrel{\eqref{eq:tube_shrinking} }{\leq} c_1 = {s}^*_{k+2|t} \stackrel{\eqref{eq:tube_bound_MPC}}{\leq}\overline{s}\stackrel{\eqref{eq:w_min_bound_bar_s}}{\leq}  \sqrt{\delta_{loc}}$. 
Thus, by using \eqref{eq:w_delta_tilde}, the disturbance $w$ satisfies ${w}_{k+1|t+1}\leq {w}^*_{k+2|t}$.
\end{enumerate}
\textbf{Part III.} State and input constraint satisfaction: \\
For $k=0,\dots,N-2$, we have  
\begin{align*}
&g_j(x_{k|t+1},u_{k|t+1})+c_js_{k|t+1}\\
\stackrel{\eqref{eq:lipschitz}\eqref{eq:prop_robust_1_2}}\leq& g_j(x^*_{k+1|t},u^*_{k+1|t})+\rho^k c_j{w}^*_{0|t}+c_js_{k|t+1}\\
\stackrel{\eqref{eq:tube_shrinking}}{\leq}&g_j(x^*_{k+1|t},u^*_{k+1|t})+c_js^*_{k+1|t}
\stackrel{\eqref{eq:real_constraints}}\leq 0.
\end{align*}
The terminal condition~\eqref{eq:terminal_region} ensures constraint satisfaction for $k=N-1$ with  
\begin{align*}
&g_j(x_{N-1|t+1},u_{N-1|t+1})+c_js_{N-1|t+1}\\
&\stackrel{\eqref{eq:lipschitz}\eqref{eq:prop_robust_1_2}\eqref{eq:tube_shrinking}}\leq g_j(x^*_{N|t},u^*_{N|t})+c_js^*_{N|t}
\stackrel{\eqref{eq:term_dist_con2}}{\leq}  0.  
\end{align*}
\textbf{Part IV.} Tube bounds \eqref{eq:tube_bound_MPC}: 
Inequalities~\eqref{eq:tube_shrinking} and \eqref{eq:disturbance_shrinking} ensure that \eqref{eq:tube_bound_MPC} hold for $k=0,\dots ,N-2$.
For $k=N-1$,  $(x^*_{N|t},s^*_{N|t})\in\mathcal{X}_f$ implies
\begin{align}
&s_{N-1|t+1}\stackrel{\eqref{eq:tube_shrinking}}{\leq}  s^*_{N|t}\stackrel{\eqref{eq:term_s}}{\leq} \overline{s},\nonumber\\
&{w}_{N-1|t+1} \stackrel{\eqref{eq:disturbance_shrinking}} \leq {w}^*_{N|t} 
= \tilde{w}_\delta(x^*_{N|t}, u^*_{N|t}, s^*_{N|t})  
\label{eq:bound_w_N}
\stackrel{\eqref{eq:w_f}}\leq     \overline{w}. 
\end{align}
\textbf{Part V.} Terminal constraint satisfaction \eqref{eq:terminal_region}: 
The terminal state and terminal tube size satisfy
\begin{align*}
&\sqrt{V_{\delta}(x_{N|t+1},x^*_{N+1|t},u^*_{N+1|t})}
 \stackrel{\eqref{eq:prop_robust_1_2}}{\leq} \rho^{N} {w}^*_{0|t},\\
&s_{N|t+1} \stackrel{\eqref{eq:tube_definition}}= \rho s_{N-1|t+1} + {w}_{N-1|t+1} 
\stackrel{\eqref{eq:tube_shrinking}, \eqref{eq:bound_w_N}}\leq \rho s^*_{N|t} - \rho^N w^*_{0|t} + w^*_{N|t}. 
\end{align*}
Thus Ass. \ref{ass:term_2},~\eqref{eq:term_dist_RPI2} ensures $(x_{N|t+1},s_{N|t+1})\in\mathcal{X}_f$ with $d_w=x^*_{N+1|t}-x_{N|t+1}$ and $w^*_{0|t}\in[\overline{w}_{\min},\overline{w}]$.\\ 
\textbf{Part VI.} 
Practical stability: For $k=0,\dots, N-1$ we have:
\begin{align*}
&\ell(x_{k|t+1},u_{k|t+1})-\ell(x^*_{k+1|t},u^*_{k+1|t})\\
\stackrel{\eqref{eq:stage_cost_bound}\eqref{eq:prop_robust_1_2}}{\leq}& \alpha_u(\rho^{k} w^*_{0|t}) 
\leq \alpha_u(\rho^{k}\overline{w}).
\end{align*}
Similar to~\cite[Prop.~5]{kohler2018novel}, continuity and exponential summability of the stage cost $\ell$ (Ass.~\ref{ass:stage_cost}) in combination with the terminal cost (Ass.~\ref{ass:term_2}) implies 
\begin{align*}
&J_N(x_{t+1},u_{\cdot|t+1})-V_N(x_t)+\ell(x_t,u_t)\\
\leq& \alpha_f(\rho^N\overline{w}/\sqrt{c_{\delta,l}})+\sum_{k=0}^{N-1} \alpha_u(\rho^{k}\overline{w})\\
\leq&\alpha_f(\overline{w}/\sqrt{c_{\delta,l}})+\alpha_{c,\rho}\left(\overline{w}\sqrt{{1}/{c_{\delta,l}}+\kappa_{\max}}\right)=:\alpha_{w}(\overline{w}),
\end{align*}
with $\alpha_w\in\mathcal{K}_{\infty}$.
Feasibility and standard arguments (c.f.~\cite[Prop.~2.16]{rawlings2017model}) imply
\begin{align*}
\alpha_{\ell}(\|x_t\|)\leq V_N(x_t)\leq& \alpha_v(\|x_t\|),\quad \alpha_v\in\mathcal{K}_{\infty},\\
V_N(x_{t+1})-V_N(x_t)\leq& -\alpha_{\ell}(\|x_t\|)+\alpha_w(\overline{w}).
\end{align*}
 Thus, the closed-loop system is practically asymptotically stable, compare~\cite[Prop.~4.3]{faulwasser2018economic}.
\end{proof}
The online constructed tube size $s$~\eqref{eq:tube_definition} in combination with the constraint tightening~\eqref{eq:real_constraints} ensures that the optimization problem~\eqref{eq:MPC_real} is recursively feasible and the constraints~\eqref{eq:constraint} are satisfied. 
The region of attraction and the bound on the maximal uncertainty $\overline{w}$ (compare Proposition~\ref{prop:term} below)  improve if a larger prediction horizon $N$ is used. 
We expect that similar theoretical properties can also be guaranteed without terminal ingredients by using a sufficiently large prediction horizon $N$ and arguments from~\cite{grune2017nonlinear,kohler2018novel}.   

\subsection{Discussion}
\label{sec:discuss}
In the following, we discuss the proposed robust MPC framework in comparison to existing robust MPC approaches (Remark~\ref{remark:dynamic}--\ref{remark:lipschitz}) and elaborate on possible extensions and applications of the proposed robust MPC framework (Remark~\ref{remark:general_cost}--\ref{remark:output_hierarchic_distributed}).
\begin{remark}
\label{remark:dynamic}
Most of the existing (nonlinear) robust MPC schemes consider  constantly bounded additive disturbances~\cite{yu2013tube,bayer2013discrete,singh2017robust}, which are a special case of the considered framework with $\tilde{w}_\delta$ constant, compare Sec.~\ref{sec:additive} and \cite{kohler2018novel}. 
While the proposed framework can utilize state and input dependent bounds on the model mismatch, we cannot directly utilize dynamic bounds on the model mismatch. 
The case of dynamic uncertainty can, e.g., be handled by imposing additional compact constraints on the input of the uncertain system in order to obtain a constant (but possibly conservative) bound on the model mismatch and thus allow for the usage of standard robust MPC methods, compare~\cite{falugi2014getting}. 
One existing approach to consider non-constant bounds on the model mismatch are so called error bounding systems as, e.g., used for MPC with reduced order models~\cite{loehning2014model}. 
A more general nonlinear robust MPC framework that can adequately handle dynamic uncertainty is still an open problem. 
\end{remark}
\begin{remark}
\label{remark:comput}
A competing (possibly less conservative) method to deal with general nonlinear robust MPC is given in~\cite{villanueva2017robust}. 
This approach includes matrices $Q_x\in\mathbb{R}^{n\times n}$, $K\in\mathbb{R}^{m\times n}$ in the online optimization that characterize a time-varying ellipsoidal tube $V_{\delta,t}$  and linear time-varying feedback $\kappa_t$ (c.f. Remark~\ref{rk:stab}), compare also~\cite{hu2018real}. 
The resulting online computational demand is comparable to a nominal MPC scheme with $(n+n^2)/2$ additional state variables $Q_x$ and $m\cdot n$ additional input/decision variables $K$, which significantly increases the computational demand for high dimensional systems, compare~\cite{hu2018real}. 
In contrast, the computational demand of the proposed scheme is comparable to a nominal MPC scheme with $n+1$ states $(x,s)$ and $m+1$ input/decision variables $(u,w)$ and thus allows for an efficient implementation to higher dimensional systems, compare the numerical example in Section~\ref{sec:num}. 
Thus, the proposed method can be interpreted as a simpler (but also more conservative) alternative to~\cite{villanueva2017robust}. 

Furthermore, there exists a strong parallel between the tube construction in~\cite{villanueva2017robust} with the predicted matrix $Q_x$,  
and uncertainty propagation in stochastic MPC approaches~\cite{hewing2019cautious} with the covariance matrix~$\Sigma_x\in\mathbb{R}^{n\times n}$. 
Based on these similarities, we expect that it is possible to devise a stochastic version of the proposed method, as a computationally efficient counterpart to~\cite{hewing2019cautious} with probabilistic guarantees, compare also the recent extension of the proposed framework to chance constraints in~\cite{schluter2020constraint}. 
\end{remark} 
\begin{remark}
\label{remark:lipschitz}
The robust MPC method in~\cite{marruedo2002input} based on a Lipschitz bound $L$ has previously been extended in~\cite{pin2009robust} to state dependent uncertainties of the form $\hat{w}=a+b\|x\|$.  
In particular, the method and results in~\cite{pin2009robust} are contained as a special case in our framework with $V_{\delta}(x,z,v)=\|x-z\|^2$, $\kappa(x,z,v)=v$ and $\rho=L$.  
Correspondingly, the design procedure and online complexity are equivalent. 
The main difference is that we use the incremental stabilizability bound $\rho$ instead of the Lipschitz bound $L$ to construct the tube. 
Thus, the resulting framework is typically significantly less conservative, compare the discussion in~\cite{kohler2018novel}, the numerical comparison in~\cite{bonzanini2019tube} and the numerical example in Section~\ref{sec:num}.
\end{remark}
\begin{remark}
\label{remark:general_cost}
The presented approach only minimizes the nominal predicted cost $\ell(x,u)$ in~\eqref{eq:MPC_real}, instead of considering the uncertainty of the prediction in the minimized cost. 
A simple and intuitive way to account for the uncertainty is to consider the following worst case ($\min-\max$) stage cost
\begin{align*}
\ell_w(z,v,s):=\ell(z,v)+\alpha_u(s)&\stackrel{\eqref{eq:stage_cost_bound}}{\geq} \max_{x: V_{\delta}(x,z,v)\leq s^2} \ell(x,\kappa(x,z,v)).
\end{align*}
This modified stage cost incentivizes cautious operation. 
In case of additive disturbances (compare Section~\ref{sec:additive}), this modification does not change the optimal open-loop trajectory. 
In~\cite{villanueva2017robust,houska2016short} an alternative cost function based on the generalized inertia (which uses matrix variables) is considered. 

Such explicit considerations of the predicted uncertainty in the cost function become even more relevant in the context of economic MPC~\cite{faulwasser2018economic}. 
A more detailed discussion on such modified stage costs based on min-max, stochastic and average cost can be found in~\cite{bayer2014tube,bayer2016robust,bayer2018optimal}. 
The stability analysis with the modified cost $\ell_w$ and the extension to such economic consideration are, however, beyond the scope of this paper. 
\end{remark}
\begin{remark}
\label{remark:output_hierarchic_distributed}
The main focus of the presented nonlinear robust MPC framework is to ensure constraint satisfaction and recursive feasibility despite disturbances/uncertainty. 
We expect that the presented ideas for robust MPC can also be used to extend existing output feedback MPC schemes~\cite{brunner2018enhancing,Kohler2019Output} and distributed/hierarchical MPC schemes~\cite{farina2012distributed} to nonlinear systems subject to state and input dependent uncertainty, which is part of future work.   
\end{remark}
%

\subsection{Offline and Online Implementation}
\label{sec:description_details}
In this section, we discuss in detail how to construct the function $\tilde{w}_\delta$ (Ass.~\ref{ass:approx}), the terminal ingredients (Ass. \ref{ass:term_2}), and summarize the overall algorithm.
\subsubsection{Nonlinear uncertainty bound}
\label{sec:non_dist_bound}
In the following we discuss some approaches to define the function $\tilde{w}_\delta$ (Ass. \ref{ass:approx}) for practical applications.
An analytical description of $\mathcal{W}$ can often be obtained based on prior knowledge of parametric uncertainty and bounds on the additive disturbance, compare the numerical examples in Sec. \ref{sec:num}. 
Alternatively, $\mathcal{W}$ can be (directly) estimated using experiments.
A general function $\tilde{w}_{\delta}$ that satisfies the posed conditions (Ass.~\ref{ass:approx}) can be computed using an appropriate parametrization and the following conceptual optimization problem
\begin{align*}
\min_{\tilde{w}_\delta}&  \int_{c=0}^{\overline{s}} \int_{r=(x,u) \in \mathcal{Z}} \tilde{w}_\delta(x,u,c)~ dr~ dc\\
\text{s.t. }&  \eqref{eq:w_delta_tilde_1},~ \eqref{eq:w_delta_tilde} \text{ hold }\forall ~(x, z,v) \in \mathbb{R}^n\times \mathcal{Z}, ~ V_\delta(x,z,v) \leq \overline{s}^2.
\end{align*}
The following proposition provides a simple way to design a function~$\tilde{w}_{\delta}$, based on continuity conditions. 
\begin{proposition}
\label{prop:w_tilde}
Let Assumptions~\ref{ass:hat_w} and~\ref{ass:contract} hold. 
There exists a function $\tilde{w}:\mathcal{Z}\rightarrow\mathbb{R}_{\geq 0}$, such that for all $(z,v)\in\mathcal{Z}$, $(z^+,v^+)\in\mathcal{Z}$, $(x,u)\in\mathbb{R}^{n+m}$ with $z^+=f(z,v)$, $d_w\in\mathcal{W}(x,u)$, the following bound is satisfied
\begin{subequations}
\begin{align}
\label{eq:w_tilde_ass}
V_{\delta}(z^+ + d_w,z^+,v^+)\leq \tilde{w}^2(x,u).
\end{align}
Assume that the following continuity condition holds for all $(x,z,v)\in\mathbb{R}^n\times\mathcal{Z}$ with $V_{\delta}(x,z,v)\leq c^2$ and any $c\in[0,\sqrt{\delta}_{loc}]$ 
\begin{align}
\label{eq:w_cont_set}
\tilde{w}(x,\kappa(x,z,v))-\tilde{w}(z,v)\leq & \tilde{\alpha}_w(c),
\end{align}
\end{subequations}
with a superadditive\footnote{%
Any function $\tilde{\alpha}_w(c)=\sum_{j=1}^{\infty} a_j c^j$ with $a_j\geq 0$ is superadditive, e.g. $\tilde{\alpha}_w$ linear or quadratic. 
} 
function $\tilde{\alpha}_w\in\mathcal{K}_{\infty}$.
Then the function
\begin{align*}
\tilde{w}_{\delta}(z,v,c):=\tilde{w}(z,v)+\tilde{\alpha}_w(c),
\end{align*}
satisfies the properties in Assumption~\ref{ass:approx}. 
\end{proposition}
\begin{proof}
Inequality~\eqref{eq:bound} ensures that condition~\eqref{eq:w_tilde_ass} is satisfied by $\tilde{w}(z,v)=\sqrt{c_{\delta,u}}\hat{w}(z,v)$, with $\hat{w}(z,v):=\max_{d_w\in\mathcal{W}(z,v)}\|d_w\|$.
Given $V_{\delta}(x,z,v)\leq c^2$ and $d_w\in\mathcal{W}(x,\kappa(x,z,v))$, we have 
\begin{align*}
\sqrt{V_{\delta}(z^+ + d_w,z^+,v^+)}\stackrel{\eqref{eq:w_tilde_ass}}{\leq}& \tilde{w}(x,\kappa(x,z,v))\\
\stackrel{\eqref{eq:w_cont_set}}{\leq}& \tilde{w}(z,v)+\tilde{\alpha}_w(c)
=\tilde{w}_{\delta}(z,v,c),
\end{align*}
which implies~\eqref{eq:w_delta_tilde_1}.
Similarly, for any $0\leq c_2\leq c_1\leq \sqrt{\delta}_{loc}$, $x\in\mathbb{R}^n$, with $V_{\delta}(x,z,v)\leq (c_1-c_2)^2$ we have
\begin{align*}
&\tilde{w}_{\delta}(x,\kappa(x,z,v),c_2)=\tilde{w}(x,\kappa(x,z,v))+\tilde{\alpha}_w(c_2)\\
\stackrel{\eqref{eq:w_cont_set}}{\leq}& \tilde{w}(z,v)+\tilde{\alpha}_w(c_1-c_2)+\tilde{\alpha}_w(c_2)
{\leq} \tilde{w}(z,v)+\tilde{\alpha}_w(c_1),  
\end{align*}
which implies~\eqref{eq:w_delta_tilde}.  
\end{proof}
\begin{corollary}
\label{corol:w_tilde}
Let Assumptions~\ref{ass:hat_w} and~\ref{ass:contract} hold. Suppose there exists a Lipschitz continuous function $\hat{w}(x,u)$ satisfying
$\hat{w}(x,u)\geq \|d_w(x,u,d)\|$ for all $(x,u,d)\in\mathcal{Z}\times\mathbb{D}$. 
Then 
\begin{align*}
\tilde{w}_{\delta}(z,v,c):=\sqrt{c_{\delta,u}}\hat{w}(z,v)+ c  {L}\sqrt{c_{\delta,u}}\sqrt{1/c_{\delta,l}+\kappa_{\max}} ,
\end{align*}
satisfies Assumption~\ref{ass:approx}, where $L$ is the Lipschitz constant of $\hat{w}$.  
\end{corollary}
\begin{proof}
The statement follows from Prop.~\ref{prop:w_tilde} using \eqref{eq:triangle_in}.  
\end{proof}
For the special case of quadratic incremental Lyapunov functions, the following proposition provides a direct method to compute $\tilde{w}_{\delta}$.  
\begin{proposition}
\label{prop:w_tilde_2}
Let Assumption~\ref{ass:hat_w} hold. 
Suppose Assumption~\ref{ass:contract} holds with\footnote{%
The proof equally applies to polytopic functions of the form $\sqrt{V_{\delta}(x,z,v)}=\max_i P_i(x-z)$, as it mainly hinges on using the triangular inequality, compare norm-like inequality in~\cite[Ass.~1]{nubert2020safe}. 
} $V_{\delta}(x,z,v)=\|x-z\|_P^2$ and $\kappa(x,z,v)=\kappa_x(x)-\kappa_x(z)+v$.  
Then Assumption~\ref{ass:approx} holds with
\begin{align}
\label{eq:w_delta_sup}
\tilde{w}_{\delta}(z,v,c):=\sup_{\{x|~\|x-z\|_P\leq c\}}~\sup_{d_w\in\mathcal{W}(x,\kappa(x,z,v))} \|d_w\|_P.
\end{align}
\end{proposition}
\begin{proof}
Condition~\eqref{eq:w_delta_tilde_1} is trivially satisfied. 
Note that the assumed structure of $\kappa$ implies 
\begin{align}
\label{eq:kappa_property}
\kappa(\tilde{x},x,\kappa(x,z,v))=&\kappa_x(\tilde{x})-\kappa_x(x)+\kappa(x,z,v)\nonumber\\
=&\kappa_x(\tilde{x})-\kappa_x(x)+\kappa_x(x)-\kappa_x(z)+v\nonumber\\
=&\kappa_x(\tilde{x})-\kappa_x(z)+v=\kappa(\tilde{x},z,v).
\end{align}
Using the quadratic nature of $V_{\delta}$ (triangular inequality) and equation~\eqref{eq:kappa_property}, the following equivalence holds
\begin{align*}
&\sup_{\{x|\|x-z\|_P\leq c_1-c_2\}}\tilde{w}_{\delta}(x,\kappa(x,z,v),c_2)\\
\stackrel{\eqref{eq:w_delta_sup}}=&\sup_{\{x|\|x-z\|_P\leq c_1-c_2\}}\sup_{\{\tilde{x}|\|\tilde{x}-x\|_P\leq c_2\}}\sup_{d_w\in\mathcal{W}(\tilde{x},\kappa(\tilde{x},x,\kappa(x,z,v)))}\|d_w\|_P\\
\stackrel{\eqref{eq:kappa_property}}{=}&  \sup_{\{x|\|x-z\|_P\leq c_1-c_2\}}\sup_{\{\tilde{x}|\|\tilde{x}-x\|_P\leq c_2\}}\sup_{d_w\in\mathcal{W}(\tilde{x},\kappa(\tilde{x},z,v))}\|d_w\|_P\\
\stackrel{\eqref{eq:w_delta_sup}}{=}&\sup_{\{\tilde{x}|\|\tilde{x}-z\|_P\leq c_1\}}\sup_{d_w\in\mathcal{W}(\tilde{x},\kappa(\tilde{x},z,v))}\|d_w\|_P=\tilde{w}_{\delta}(z,v,c_1),
\end{align*}
compare Fig.~\ref{fig:monotonicity} for an illustration. 
This condition implies~\eqref{eq:w_delta_tilde} and thus satisfaction of Assumption~\ref{ass:approx}.  
\end{proof}
The method in Proposition~\ref{prop:w_tilde_2} is not applicable to arbitrary nonlinear incremental Lyapunov functions $V_{\delta}$ and feedbacks $\kappa$, since the monotonicity property~\eqref{eq:w_delta_tilde} does not necessarily hold.

Propositions~\ref{prop:contract}--\ref{prop:w_tilde} (and Corollary~\ref{corol:w_tilde}) provide simple procedures to compute $\tilde{w}_{\delta}$ and $c_j$, which only use scalar bounds describing the incremental stabilizability property (Ass.~\ref{ass:contract}) and do not explicitly use the incremental Lyapunov function $V_{\delta}$. 
In case of quadratic incremental Lyapunov functions $V_{\delta}$, Proposition~\ref{prop:w_tilde_2} provides a procedure that uses the shape of $V_{\delta}$ and $\mathcal{W}$ to compute the least conservative function $\tilde{w}_{\delta}$ satisfying Assumption~\ref{ass:approx}. 
Furthermore, in Section~\ref{sec:param} we consider polytopic incremental Lyapunov functions $V_{\delta}$ for LPV systems
and provide simple (linear) formulas for $\tilde{w}_{\delta},~c_j$ that exploit the shape of $V_{\delta},~\mathcal{W}$.
The direct construction of $\tilde{w}_{\delta}$ (using Prop.~\ref{prop:w_tilde}) is illustrated in Section~\ref{sec:num} with a numerical example. 
In general there exists a degree of freedom in the design of $\tilde{w}_{\delta}$ that allows for a trade-off between conservatism and computational complexity, compare also Remark~\ref{rk:lpv_QP} and the numerical example (Sec.~\ref{sec:num}).

%
\subsubsection{Terminal ingredients}
\label{sec:terminal_ing}
In the following, we illustrate a simple procedure to compute terminal ingredients $(V_f,~\mathcal{X}_f,~k_f,~\overline{w})$ that satisfy Assumption~\ref{ass:term_2}.
\begin{proposition}
\label{prop:term}
Let Assumptions~\ref{ass:hat_w}--\ref{ass:approx} hold. 
Assume that there exist a feedback $k_f$ and a terminal cost $V_f$, such that   the conditions~\eqref{eq:term_dist_dec2},~\eqref{eq:term_contin2} in Assumption~\ref{ass:term_2} hold for all $x\in\mathbb{R}^n$ with $V_f(x)\leq \overline{\gamma}$ and some constant $\overline{\gamma}>0$. 
Without loss of generality\footnote{%
\cite[Lemma B.1]{jiang2001input} There exists a function $\hat{\alpha}\in\mathcal{K}_{\infty}$ such that $\hat{\alpha}\leq\alpha_l\circ \alpha_f^{-1}$, $id-\hat{\alpha}\in\mathcal{K}$. 
Thus, we can replace $\alpha_l\circ\alpha_f^{-1}$ by $\hat{\alpha}$ in the proof and~\eqref{eq:termProp_5}--\eqref{eq:N_0}. 
}, suppose that $id-\alpha_l\circ \alpha_f^{-1}\in\mathcal{K}$, where $id$ denotes the identity. 
Consider the terminal set
\begin{align*}
\mathcal{X}_f=\{(x,s)\in\mathbb{R}^{n+1}|~s\in[0,\overline{s}_f],~V_f(x)\leq \gamma\},
\end{align*}
  with some positive constants $\overline{s}_f,~\gamma$.
There exist functions $f_j\in\mathcal{K}_{\infty}$, $j=1,\dots,p$, such that the following condition holds for all $\gamma\in[0,\overline{\gamma}]$
\begin{align}
\label{eq:f_j}
\sup_{(x,s)\in\mathcal{X}_f}g_j(x,k_f(x))+c_js\leq g_j(0,0)+ f_j({\gamma})+c_j\overline{s}_f.
\end{align}
Suppose that the function $\tilde{w}_{\delta}$ and the terminal set satisfy
\begin{align}
\label{eq:termProp_2}
\sup_{(x,s)\in\mathcal{X}_f}\tilde{w}_{\delta}(x,k_f(x),s)\leq a_0+ a_1\sqrt{\gamma}+a_2\overline{s}_f,
\end{align}
with some positive constants $a_0,~a_1,~a_2$ and $a_2<1-\rho$. 
If there exists a constant $\gamma\in(0,\overline{\gamma}]$, that satisfies 
\begin{align}
\label{eq:termProp_3}
\dfrac{a_0+a_1\sqrt{\gamma}}{1-\rho-a_2}\leq \min_{j=1,\dots,p}\dfrac{-g_j(0,0)-f_j(\gamma)}{c_j}\leq \overline{s},
\end{align}
then Assumption~\ref{ass:term_2} is satisfied with $k_f,~V_f,~\mathcal{X}_f,~\gamma$, $N\geq N_0$,
\begin{subequations}
\label{eq:termProp_4_5_N_0}
\begin{align}
\label{eq:termProp_4}
\overline{s}_f:=&\min_{j=1,\dots,p}\dfrac{-g_j(0,0)-f_j(\gamma)}{c_j},\\
\label{eq:termProp_5}
\overline{w}:=&\rho^{-N}\cdot\sqrt{c_{\delta,l}}{\alpha_f^{-1}(\alpha_l(\alpha_f^{-1}(\gamma)))},\\
\label{eq:N_0}
N_0:=&\log_{\rho}\left(\dfrac{\sqrt{c_{\delta,l}}\alpha_f^{-1}(\alpha_l(\alpha_f^{-1}(\gamma)))}{(1-\rho)\overline{s}_f}\right).
\end{align}
\end{subequations}
\end{proposition}
\begin{proof}
\textbf{Part I. }
The robust positive invariance condition~\eqref{eq:term_dist_RPI2} follows from
\begin{align*}
&V_f(x^+ + d_w)\stackrel{\eqref{eq:term_contin2}}{\leq}  V_f(x^+ )+\alpha_f(\|d_w\|)\\
\stackrel{\eqref{eq:term_dist_dec2},\eqref{eq:bound}}{\leq}& V_f(x)-\ell(x,k_f(x))+\alpha_f(\rho^N\overline{w}/\sqrt{c_{\delta,l}})\\
\stackrel{\eqref{eq:stage_cost_l}}{\leq}& V_f(x) -\alpha_l(\|x\|)+\alpha_f(\rho^N\overline{w}/\sqrt{c_{\delta,l}})\\
\stackrel{\eqref{eq:term_contin2}}{\leq} &V_f(x)-\alpha_l(\alpha_f^{-1}(V_f(x)))+\alpha_f(\rho^N\overline{w}/\sqrt{c_{\delta,l}})\\
\leq& \gamma-\alpha_l(\alpha_f^{-1}(\gamma))+\alpha_f(\rho^N\overline{w}/\sqrt{c_{\delta,l}}) \stackrel{\eqref{eq:termProp_5}}{=}  \gamma
\end{align*}
and
\begin{align*}
&s^+\leq \rho s-\rho^N w+\tilde{w}_{\delta}(x,k_f(x),s)\\
\stackrel{\eqref{eq:termProp_2}}{\leq}& \rho \overline{s}_f+a_0+a_1\sqrt{\gamma}+a_2\overline{s}_f
\stackrel{\eqref{eq:termProp_3}\eqref{eq:termProp_4}}{\leq} \overline{s}_f.
\end{align*}
\textbf{Part II. } Assumptions~\ref{ass:stage_cost}, \ref{ass:gen_nonlin_con_Lipschitz} and \eqref{eq:term_dist_dec2} ensure
\begin{align*}
&g_j(x,k_f(x))-g_j(0,0)\stackrel{\eqref{eq:nonlin_con_Lipschitz}}{\leq}L_j\|(x,k_f(x))\|\\
\stackrel{\eqref{eq:stage_cost_l}}{\leq}& L_j\alpha_l^{-1}(\ell(x,k_f(x)) \stackrel{\eqref{eq:term_dist_dec2}}{\leq}L_j \alpha_l^{-1}(V_f(x))\leq L_j\alpha_l^{-1}(\gamma),
\end{align*}
for all $(x,s)\in\mathcal{X}_f$.
Thus, \eqref{eq:f_j} is satisfied with $f_j(\gamma)=L_j\alpha_l^{-1}(\gamma)$. 
Satisfaction of~\eqref{eq:term_dist_con2} follows from~\eqref{eq:f_j} and~\eqref{eq:termProp_4}. \\
\textbf{Part III. } 
Satisfaction of condition~\eqref{eq:w_f} follows from
\begin{align*}
&\tilde{w}_{\delta}(x,k_f(x),s)\stackrel{\eqref{eq:termProp_2}}{\leq} a_0+a_1\sqrt{\gamma}+a_2\overline{s}_f\\
\stackrel{\eqref{eq:termProp_3}\eqref{eq:termProp_4}}{\leq }& (1-\rho)\overline{s}_f
\stackrel{\eqref{eq:termProp_5}\eqref{eq:N_0}}{=}\rho^{N-N_0}\overline{w}\leq \overline{w}.
\end{align*}
Similarly, \eqref{eq:termProp_4} in combination with \eqref{eq:termProp_3} implies \eqref{eq:term_s}. 
\end{proof}
This proposition provides a simple procedure to compute a suitable terminal set under few conditions for relatively general nonlinear systems. 
For a quadratic stage cost $\ell(x,u) = \|x\|_Q^2 + \|u\|_R^2$, a quadratic terminal cost $V_f(x) = \|x\|_P^2$ and a linear feedback $k_f(x)=K_fx$ that locally satisfy \eqref{eq:term_dist_dec2}, \eqref{eq:term_contin2}, can be computed using standard methods \cite{rawlings2017model,chen1998quasi}. 
Condition~\eqref{eq:termProp_2} requires a bound on the uncertainty close to the origin. 
The maximal tube size $\overline{s}_f$ is defined as large as possible, such that the tightened constraints~\eqref{eq:term_dist_con2} are satisfied in the terminal set. 
In Proposition~\ref{prop:LPV_simple} below, we show how these conditions simplify for LPV systems. 
Furthermore, given an additional norm-like inequality on $V_\delta$, simple design for the terminal region directly using $V_\delta$ can be found in the recent papers \cite[Prop.~1]{nubert2020safe} and \cite[Prop.~6]{kohler2019robust}.

Interestingly, equation~\eqref{eq:termProp_5} shows that we can consider arbitrary large $\overline{w}$, if $N$ is sufficiently large. 
The practical implication is that we can even operate the system in regions with large uncertainty, if we can be certain that this effects only an initial part of a sufficiently long predicted trajectory ($N\texttt{>>}1$). In particular, the constraints \eqref{eq:tube_definition}, \eqref{eq:tube_bound_MPC} and \eqref{eq:terminal_region} implicitly define the following bound on the uncertainty $w_{k|t}$
\begin{align*}
s_{N|t} = \sum_{k=0}^{N-1}\rho^{N-1-k}w_{k|t}\leq \overline{s}_f \leq \overline{s}. 
\end{align*}


\subsubsection{Algorithm}
\label{sec:algorithm}
The offline and online computation are summarized in Algorithm \ref{alg:offline} and \ref{alg:online}, respectively.
\begin{algorithm}[H]
\caption{Offline Computation}
\label{alg:offline}
\begin{algorithmic}[1]
\State Choose stage cost $\ell$ (Assumption~\ref{ass:stage_cost}) and constraint set $\mathcal{Z}$.
\State Verify incremental stabilizability (Ass.~\ref{ass:contract},~\cite[Alg.~2]{koehler2020nonlinear}):
\Statex - Obtain scalars $\rho\in [0,1)$, $\delta_{loc}$, compute $c_j$ (Prop.~\ref{prop:contract}).
\State Obtain uncertainty bound $\tilde{w}_{\delta}$ (Ass.~\ref{ass:approx}):
\Statex - Describe uncertainty $\mathcal{W}$ (Sec. \ref{sec:non_dist_bound}, Ass.~\ref{ass:hat_w}).  
\Statex - Compute $\tilde{w}_{\delta}$ using Prop.~\ref{prop:w_tilde} or \ref{prop:w_tilde_2}.
\State Compute terminal ingredients (Ass. \ref{ass:term_2}, Prop.~\ref{prop:term}): 
\Statex - Compute $V_f,~k_f$ \cite{chen1998quasi}.  
\Statex - Choose $\gamma\in(0,\overline{\gamma}]$, such that condition~\eqref{eq:termProp_3} holds. 
\Statex - Compute $\overline{s}_f,~\overline{w},~N_0$~\eqref{eq:termProp_4_5_N_0}. 
\Statex - Choose $N\geq N_0$.
\end{algorithmic}
\end{algorithm}
\begin{algorithm}[H]
\caption{Online Computation}
\label{alg:online}
\begin{algorithmic}[1]
\State Measure the state $x_t$.
\State Solve the MPC optimization problem \eqref{eq:MPC_real}.
\State Apply the control input: $u_t = u^*_{0|t}$.
\State Set $t=t+1$ and go back to 1.
\end{algorithmic}
\end{algorithm}
%
%

\section{Special cases} 
\label{sec:RMPC_special}
This section considers important special cases of the proposed framework. 
The special case of additive disturbances is presented in Section~\ref{sec:additive}.  
Section~\ref{sec:param} considers LPV systems and compares the resulting scheme with some of the existing methods. 
\subsection{Additive disturbances} 
\label{sec:additive}
The simplest and thus most common way to treat uncertainty in MPC is to consider additive disturbances with a constant bound $\mathcal{W}$, as for example done in~\cite{chisci2001systems,mayne2005robust,marruedo2002input,yu2013tube,bayer2013discrete,singh2017robust,yu2010robust}. 
This is a special case of the previous derivation, by considering the constant 
\begin{align}
\label{eq:w_max_add}
\overline{w}_{\min}=&\overline{w}:=\sup_{z,v,v^+,d_w}\sqrt{V_{\delta}(f(z,v)+d_w,f(z,v),v^+)}\\
\text{s.t. }& (z,v)\in\mathcal{Z},~(f(z,v),v^+)\in\mathcal{Z},~d_w\in\mathcal{W}(z,v). \nonumber
\end{align}
In the following, we show how the problem simplifies in this case. 
A preliminary\footnote{%
Compared to~\cite{kohler2018novel}, a more general stage cost and nonlinear constraints are considered. 
Furthermore, \cite{kohler2018novel} considers no terminal ingredients. 
} version of the following robust MPC scheme has been presented in~\cite{kohler2018novel}. 
\begin{assumption}
\label{ass:term} 
There exist a terminal controller $k_f:\mathbb{R}^n\rightarrow\mathbb{R}^m$, a terminal cost function $V_f:\mathbb{R}^n\rightarrow \mathbb{R}_{\geq 0}$, and a terminal set $\mathcal{X}_f\subset\mathbb{R}^n$, such that the following properties hold for any $x\in\mathcal{X}_f$ and all $d_w\in\mathbb{R}^n$, such that $V_{\delta}(x^+ + d_w,x^+,k_f(x^+))\leq \rho^{2N}\overline{w}^2$ with $x^+=f(x,k_f(x))$:
\begin{subequations}
\label{eq:term_dist}
\begin{align}
\label{eq:term_dist_dec}
V_f(x^+)\leq& V_f(x)-\ell(x,k_f(x)),\\
\label{eq:term_dist_RPI}
x^+ + d_w\in & \mathcal{X}_f,\\
\label{eq:w_bar_additive}
\dfrac{1-\rho^N}{1-\rho}\bar{w}\leq& \overline{s}= \sqrt{\delta_{\text{loc}}},\\
\label{eq:term_dist_con}
g_j(x,k_f(x))+{c}_{j}\dfrac{1-\rho^N}{1-\rho}\overline{w}\leq& 0, \quad  j = 1, \dots, p,
\end{align}
\end{subequations}
Furthermore, there exists a function $\alpha_f\in\mathcal{K}_{\infty}$ such that
\begin{align*}
V_f(z)\leq V_f(x)+\alpha_f(\|x-z\|),\quad \forall x,z\in\mathcal{X}_f.
\end{align*}
\end{assumption}
Compared to  Assumption~\ref{ass:term_2}, the conditions~\eqref{eq:w_f}-\eqref{eq:term_dist_con2} simplify to~\eqref{eq:w_bar_additive}--\eqref{eq:term_dist_con}, with ${s}=\frac{1-\rho^N}{1-\rho}\overline{w}$,~$\overline{w}_{\min}=\overline{w}$.   
Compared to the nominal conditions for the terminal ingredients~\cite{chen1998quasi,rawlings2017model}, the tightened  state and input constraints~\eqref{eq:term_dist_con} need to be satisfied and the terminal set should be robust positively invariant (RPI)~\eqref{eq:term_dist_RPI}. 
The corresponding optimization problem is given by
\begin{subequations}
\label{eq:MPC}
\begin{align}
V_N(x_t)&=\min_{u_{\cdot|t},x_{\cdot|t}} J_N(x_{\cdot | t},u_{\cdot|t})\\
\text{s.t. }
&x_{0|t}=x_t, ~ x_{k+1|t}=f(x_{k|t},u_{k|t}),\\
\label{eq:MPC_con}
&g_j(x_{k|t},u_{k|t})+\dfrac{1-\rho^k}{1-\rho}{c}_{j}\overline{w}\leq 0,\\
\label{eq:MPC_term}
&x_{N|t}\in\mathcal{X}_f,~ k=0,\dots, N-1,\quad j=1,\dots p.
\end{align}
\end{subequations}
The only difference compared to a nominal MPC scheme is the fact that the tightened constraints~\eqref{eq:MPC_con} are considered along the prediction horizon. 
Compared to the formulation in~\eqref{eq:MPC_real}, the constraint tightening is computed offline, which leads to a computational complexity equivalent to the corresponding nominal MPC scheme. 
However, describing the uncertainty as additive disturbances instead of state and input dependent uncertainty can introduce a lot of conservatism. 
The following theorem establishes the corresponding closed-loop properties. 
\begin{theorem}
\label{thm:robust_simple}
Let Assumptions~\ref{ass:hat_w}--\ref{ass:approx} and \ref{ass:term} hold, and suppose that problem~\eqref{eq:MPC} is feasible at time $t=0$. 
Then~\eqref{eq:MPC} is recursively feasible, the constraints~\eqref{eq:constraint} are satisfied and the origin is practically asymptotically stable for the resulting closed-loop system~\eqref{eq:close}.
\end{theorem}
\begin{proof}
The proof is similar to~\cite{kohler2018novel} and is a special case of Theorem~\ref{thm:main}. 
In particular, $\overline{w}_{k|t}=\overline{w}$ implies 
\begin{align*}
s_{k|t}=\sum_{i=0}^{k-1}\rho^i\overline{w} =\dfrac{1-\rho^k}{1-\rho}\overline{w}.
\end{align*}
The remainder of the proof is the same as in Theorem~\ref{thm:main}. 
\end{proof}
\begin{remark}
\label{remark:linear}
In~\cite{chisci2001systems}, a constraint tightening for linear systems subject to additive disturbances is considered. 
The considered constraint tightening~\eqref{eq:MPC_con} can be interpreted as an overapproximation of the constraint tightening in~\cite{chisci2001systems}, compare~\cite[Remark~14]{kohler2018novel} for details. 
In particular, the overapproximation is a result of the simple description of the stabilizability property using inequality~\eqref{eq:contract}, instead of explicitly using the general nonlinear uncertain dynamics~\eqref{eq:sys_w}. 
Thus, the presented framework can be viewed as a simple to implement (cf.~\cite{wabersich2018safe}) extension of~\cite{chisci2001systems} to nonlinear uncertain systems, based on the incremental stabilizability property (Ass.~\ref{ass:contract}). 
In ~\cite{mesbah2018backoff} a more elaborate method to compute an explicit offline constraint tightening for nonlinear uncertain systems is proposed. 
The resulting constraint tightening often has a similar shape as~\eqref{eq:MPC_con}, compare~\cite[Fig.~2]{mesbah2018backoff}, however, the method suffers from a more complex offline computation. 
\end{remark}

 %
\subsection{Linear parameter varying systems}
\label{sec:param}
We consider the special case of LPV systems
\begin{align}
\label{eq:LPV}
x_{t+1}=&(A_0+A_{\theta_t}) x_t+(B_0+B_{\theta_t}) u_t+Ed_t,
\end{align}
where $d_t$ are additive disturbances and $A_\theta \in \mathbb{R}^{n\times n}$, $B_\theta \in \mathbb{R}^{n\times m}$ are time-varying uncertain matrices
\begin{align}
\label{eq:A_theta}
 A_{\theta_t}=&\sum_{i=1}^q \theta_{i,t} A_i,\quad B_{\theta_t}=\sum_{i=1}^q \theta_{i,t} B_i,
\end{align}
with time-varying unknown\footnote{%
In case $\theta_t$ has a bounded rate of change and $\theta_t$ can be measured/estimated online, updating the nominal system $f$ can reduce the conservatism, compare e.g. MPC for LPV systems~\cite{abbas2019tube} and robust adaptive MPC~\cite{kohler2019robust}. 
} parameters $\theta_t\in \mathbb{R}^q$. 
We assume that there exist polytopes $\mathbb{D},~\Theta$, such that $\theta_t\in{\Theta}$,~$d_t\in\mathbb{D}$ for all $t\geq 0$. 
The nominal prediction model $f$ is simply a linear time-invariant (LTI) system
\begin{align}
\label{eq:sys_linear}
f(x_t,u_t)=A_0x_t+B_0u_t.
\end{align}
Furthermore, we consider a polytopic constraint set $\mathcal{Z}$~\eqref{eq:constraint} with
\begin{align}
\label{eq:con_polytope}
g_j(x,u)=L_{j,x}x+L_{j,u}u-1\leq 0, \quad j=1,\dots,p,
\end{align}
and a quadratic stage cost $\ell(x,u)=\|x\|_Q^2+\|u\|_R^2$, with $Q,~R$ positive definite. 
The following proposition shows how the design simplifies in this setup.  
\begin{proposition}
\label{prop:LPV_simple}
Suppose $(A_0, B_0)$ is stabilizable, i.e. there exists a feedback $K$ such that $A_K:=A_0+B_0K$ is Schur.
\begin{enumerate}
\item 
There exists a compact polytope\footnote{%
In case of symmetric polytopes $\mathcal{P}$, the formulas simplify with $V_{\delta}(x,z,v)=\|P(x-z)\|_{\infty}$ and only half the indices $i$ need to be enumerated in $\tilde{w}(z,v)$. } $\mathcal{P} = \{x\in\mathbb{R}^n|~P_ix\leq 1,~i=1,\dots r\}$ with $A_K \mathcal{P} \subseteq \rho \mathcal{P},~\rho \in [0,1)$. 
Assumption~\ref{ass:contract} is satisfied with $\sqrt{V_\delta(x,z,v)} = \max_i P_i(x-z)$, $\kappa(x,z,v)=v+K(x-z)$. 
Furthermore, constants $c_j$ satisfying \eqref{eq:lipschitz} in Prop.~\ref{prop:contract} can be computed using~\eqref{eq:c_j_polytope}.   
\item Denote the vertices $\theta^{j}\in\text{vert}(\Theta)$.  
The following function satisfies Assumption~\ref{ass:approx}
\begin{align}
\label{eq:w_LPV_poly}
\tilde{w}_{\delta}(z,v,c)&:=\tilde{w}(z,v)+L_{w}c,\\
\tilde{w}(z,v)&:=\max_{i,j} P_i[A_{\theta^{j}}z+B_{\theta^{j}}v]+\overline{d}_i,\nonumber\\
L_w &:=  \max_{i,j}\max_{\Delta x\in \mathcal{P}}  P_i(A_{\theta^j}+B_{\theta^j}K)\Delta x,\nonumber\\
\overline{d}_i&:=\max_{d\in\mathbb{D}} P_iEd,~i=1,\dots,r.
\end{align}
\item  Denote $c_{max}:=\max_{j=1,\dots,p}c_{j}$, $\overline{d}=\max_{i=1,\dots,r}\overline{d}_i$ and assume further that
\begin{align}
\label{eq:poly_alpha_suff_2}
\rho+L_{w}+c_{\max}\overline{d}\leq 1. 
\end{align}
Then Assumption~\ref{ass:term_2} is satisfied with $\overline{w}=\overline{s}=\infty$, 
\begin{subequations}
\begin{align}
&k_f(x)=Kx,~V_f(x)=\|x\|_{P_f}^2,\\
\label{eq:Lyap_disc}
&P_f=A_K^\top P_f A_K+Q+K^\top RK,\\
\label{eq:poly_term_set_new}
&\mathcal{X}_f:=\{(x,s)|P_ix+s\leq 1/c_{\max},~i=1,\dots,r\}.
\end{align}
\end{subequations}
\end{enumerate}
\end{proposition}
\begin{proof}
\textbf{Part I. }
The existence of the contractive polytopic set $\mathcal{P}$ is ensured by \cite[Lemma~5]{rakovic2007minkowski}. Furthermore, $\sqrt{V_\delta}$ is the Minkowski functional of $\mathcal{P}$ \cite[Definition~3.2]{blanchini1999set}, and thus satisfies \eqref{eq:contract} with $\delta_{loc}$ (and thus $\overline{s})$ arbitrarily large.
Compactness of $\mathcal{P}$ implies that there exists a constant $c_{\delta,l}$, such that~\eqref{eq:bound} holds. 
The conditions~\eqref{eq:bound}, \eqref{eq:k_max} are satisfied with $c_{\delta,u}:=\max_i\|P_i\|^2$ and $\kappa_{\max} := \|K\|^2/c_{\delta,l}$.
The constants $c_j$ can be determined by the  following linear program (LP)
\begin{align}
\label{eq:c_j_polytope}
c_j=\max_{x\in\mathcal{P}} [L_{j,x}+L_{j,u}K]x.
\end{align}
The contraction $\rho$ is computed using $\min \rho~\text{s.t. } A_K\mathcal{P}\subseteq\rho\mathcal{P}$, which can be cast as an LP, compare \cite[Thm.~4.1]{blanchini1999set}.\\ 
\textbf{Part II. }
We show the second claim by using Proposition \ref{prop:w_tilde}.
Condition \eqref{eq:w_tilde_ass} is satisfied with
\begin{align*}
&\sqrt{V_\delta(z^++d_w, z^+,v^+)} = \max_i P_i \left[ A_{\theta} z + B_{\theta} v + Ed\right] \\
&\leq \max_{i,j} P_i[A_{\theta^{j}}z+B_{\theta^{j}}v]+\overline{d}_i=\tilde{w}(z,v).
\end{align*}
The inequality follows from the fact that $\sqrt{V_{\delta}}$ is linear in $\theta,$ and $d$, and thus attains its extreme value on a vertex. 
Similarly, the continuity condition in \eqref{eq:w_cont_set} is satisfied with $\sqrt{V_{\delta}(x,z,v)}=\max_i P_i\Delta x\leq c$, $\Delta x=x-z$ and
\begin{align}
\label{eq:w_L_theta_LPV}
\nonumber &\tilde{w}(z+\Delta x, v+ K \Delta x) \\
\nonumber=&\max_{i,j} P_i[A_{\theta^{j}}(z+\Delta x)+B_{\theta^{j}}(v+K\Delta x)]+\overline{d}_i\\
\nonumber= &\max_{i,j}P_i[A_{\theta^{j}}z+B_{\theta^{j}}v]+\overline{d}_i
+P_i(A_{\theta^j}+B_{\theta^j}K)\Delta x\\
\stackrel{\eqref{eq:w_LPV_poly}}{\leq} &\tilde{w}(z,v)+L_{w}c.
\end{align}
\textbf{Part III. }
The proof is similar to Prop.~\ref{prop:term}, but utilizes the the incremental Lyapunov function $V_{\delta}$ to construct the terminal region, which results in simpler conditions. 
A similar design has been proposed in~\cite{Koehler2019Adaptive}.
Conditions~\eqref{eq:term_dist_dec2},~\eqref{eq:term_contin2} are satisfied with $P_f$ based on the discrete-time Lyapunov equation~\eqref{eq:Lyap_disc}. 
Conditions~\eqref{eq:w_f}, \eqref{eq:term_s} are satisfied by definition, with $\overline{s}=\overline{w}=\infty$. 
The uncertainty satisfies
\begin{align}
\label{eq:poly_termProp_2}
&\max_{(x,s)\in\mathcal{X}_f}\tilde{w}_{\delta}(x,k_f(x),s)
=\max_{P_ix+s\leq 1/c_{\max}}\tilde{w}(x,k_f(x))+L_{w}s\nonumber\\
\stackrel{\eqref{eq:w_L_theta_LPV}}{\leq}&\max_{P_ix+s\leq 1/c_{\max}}\tilde{w}(0,0)+L_{w}({s}+\max_i P_ix)=\overline{d}+L_w/c_{\max}. 
\end{align}
For any $(x,s)\in\mathcal{X}_f$, the robust positive invariance condition~\eqref{eq:term_dist_RPI2} follows from
\begin{align}
\label{eq:term_linear3}
&s^+ + \max_i P_i(x^+ + d_w)\\
\leq& s^++\max_i P_i x^++\max_i P_i d_w\nonumber\\
\leq&\rho s-\rho^N w+\tilde{w}_{\delta}(x,k_f(x),s)+ \rho\max_i P_i x+\rho^N w\nonumber\\
\stackrel{\eqref{eq:poly_termProp_2}}{\leq} &\rho(s+\max_i P_i x)+\overline{d}+L_{w}/c_{\max}
\stackrel{\eqref{eq:poly_alpha_suff_2}}{\leq} 1/c_{\max}.\nonumber
\nonumber
\end{align}
Satisfaction of~\eqref{eq:term_dist_con2} follows from 
\begin{align*}
&g_j(x,k_f(x))+c_js=-1+(L_{j,x}+L_{j,u}K)x+c_js\\
\stackrel{\eqref{eq:lipschitz}}{\leq}& -1+c_{\max}({s}+\max_i P_ix)\stackrel{\eqref{eq:poly_term_set_new}}{\leq }0,
\end{align*}
with $c_j$ from~\eqref{eq:c_j_polytope}. 
\end{proof}
\begin{remark}
\label{rk:lpv_QP}
The function $\tilde{w}_{\delta}$ is such that the corresponding robust MPC optimization problem~\eqref{eq:MPC_real} is a quadratic program (QP), where the constraint~\eqref{eq:w_definition} can be formulated as $r\cdot p$ linear inequality constraints. 
The computational demand can be decreased by considering a simpler (but more conservative) function $\tilde{w}_{\delta}$. For example, with $\Theta=[-1,1]^q$ we can use
\begin{align}
\label{eq:poly_w_delta_simple}
\tilde{w}_{\delta}(z,v,c)&= \max_iP_i \sum_{j=1}^q (A_j z + B_j v)+ \overline{d}_i + L_w c ,\\
L_w &= \max_i\max_{\Delta x\in \mathcal{P}} P_i \sum_{j=1}^q(A_j+B_jK)\Delta x, \nonumber
\end{align}
which can be implemented with $r$ linear inequality constraints. 
Thus, in general the design of $\tilde{w}_{\delta}$ includes a degree of freedom that can be utilized to trade off computational complexity against conservatism, compare also~\cite{Koehler2019Adaptive}. 
The results in Proposition~\ref{prop:LPV_simple} can also be formulated with a quadratic incremental Lyapunov function $V_{\delta}(x,z,v)=\|x-z\|_P^2$, however, the construction of $\tilde{w}_{\delta}$ then typically requires either conservative overapproximations or results in a quadratically constrained QP (QCQP). 
\end{remark}
\begin{remark}
\label{remark:L_theta}
The construction of the terminal ingredients uses condition~\eqref{eq:poly_alpha_suff_2}, which requires $\rho+L_{w}<1$ in addition to a bound on the additive disturbances $\overline{d}$ w.r.t. the size of the constraint set (as characterized by $c_{\max}$). 
The condition $\rho+L_{w}<1$ is a natural condition for LPV systems in the considered setup, as
it ensures that $A_\theta+B_\theta K$ is stable with the common Lyapunov function $V=\max_i P_ix$. 
The definition of $\tilde{w}_\delta$ in Proposition \ref{prop:LPV_simple} implies
\begin{align*}
s_{N|t} \geq  \sum_{j=0}^{N-1} (L_w+\rho)^{N-j-1} \tilde{w}(x_{j|t}, u_{j|t}). 
\end{align*}
Thus, $\rho+L_{w}<1$ also ensures that the tube size $s$ does not increase arbitrarily. 
This is in contrast to the competing approaches~\cite{marruedo2002input,limon2005robust,pin2009robust}, where the tube size $s$ can grow exponentially with the prediction horizon $N$, compare the numerical example in Section~\ref{sec:num}. 
\end{remark}
\begin{remark}
\label{rk:lpv_nonlin_adaptive}
The presented design can be also applied to nonlinear systems affine in parameters $\theta$ with quadratic or polytopic incremental Lyapunov functions $V_{\delta}$ (Prop.~\ref{prop:w_tilde_2}), compare also the numerical example (Sec.~\ref{sec:num}). 
The simple condition~\eqref{eq:poly_alpha_suff_2} is, however, not necessarily suitable for more general incremental Lyapunov functions $V_{\delta}$, since the proof exploits the triangular inequality in~\eqref{eq:term_linear3}.  
Recently, the presented method has been extended to adaptive MPC with online estimated parameters $\theta$ in~\cite{kohler2019robust}.
\end{remark}

\begin{remark}
\label{rk:lpv_comparison}
For the special case of LPV systems, the proposed approach can be viewed as a simplified (and hence computationally more efficient) version of existing methods~\cite{fleming2015robust,houska2016short}, \cite{feng2019min}, \cite[Chap.~5]{kouvaritakis2016model}.   
In particular, with a polytopic tube $\sqrt{V_{\delta}(x,z,v)}=\max_i P_i(x-z)$ the approach is similar to~\cite{fleming2015robust}.
The main simplification is that we characterize the tube and contractivity of this polytope with scalars $w,\rho$, instead of treating each vertex/halfspace of the polytope separately. 

Similarly, in case of ellipsoidal tubes $V_{\delta}(x,z,v)=\|x-z\|_P^2$, we can formulate a simplified version of~\cite{houska2016short}. 
Here, the simplification (and thus reduction in the computational demand) is more pronounced, as we use a scalar $s$ to parameterize the tube instead of optimizing matrices $P\in\mathbb{R}^{n\times n}$ online, compare also Remark~\ref{remark:comput}. 

While these approximations can introduce some conservatism, the resulting simplicity is also the main benefit of the proposed approach. 
In particular, the considered formulation is equally applicable to linear and nonlinear systems, additive disturbance, parametric uncertainty and general nonlinear mixed uncertainty. 
The description with $\tilde{w}_{\delta}$ makes it also easy to use further approximations, such as~\eqref{eq:poly_w_delta_simple}, to reduce the complexity and thus allow an application to higher dimensional systems.
The conservatism and computational complexity will be further explored in the numerical example in Section~\ref{sec:num}. 
A quantitative comparison for linear systems with a polytopic tube can be found in~\cite{Koehler2019Adaptive}. 
\end{remark}

\section{Case study: nonlinear quadrotor}
\label{sec:num}
The following example details how the proposed approach can also be applied to uncertain nonlinear \textit{continuous time} system, compare Appendix~\ref{app:cont}). 
Furthermore, we demonstrate the computational efficiency and performance relative to existing approaches using min-max differential inequalities~\cite{villanueva2017robust,hu2018real} and Lipschitz bounds~\cite{marruedo2002input,pin2009robust}.
In addition, we showcase that considering state and input dependent uncertainty reduces the conservatism compared to simple constant bounds on the uncertainty.
The offline and online computation is done using SeDuMi-1.3~\cite{sturm1999using} and CasADi~\cite{andersson2019casadi}, respectively. 
An additional example with robust collision avoidance for autonomous vehicles using a non-quadratic incremental Lyapunov function (c.f.~\cite{koehler2020nonlinear}) can be found in~\cite{Soloperto2019Collision}. 
A recent experimental implementation for robust collision avoidance in robotics can be found in~\cite{nubert2020safe}. 
\subsubsection*{System model}
We consider the following continuous-time 10-state quadrotor model \cite{hu2018real}
\begin{align*}
\dot{x}_1=&v_1+0.1{w}_1,\quad \dot{v}_1=g\tan(\phi_1),\\
\dot{x}_2=&v_2+0.3{w}_2,\quad \dot{v}_2=g\tan(\phi_2),\\
\dot{x}_3=&v_3+0.5w_3,\quad \dot{v}_3=-g+k_Tu_3,\\
\dot{\phi}_1=&-d_1\phi_1+\omega_1,\quad \dot{\omega}_1=-d_0\phi_1+n_0u_1,\\
\dot{\phi}_2=&-d_1\phi_2+\omega_2,\quad \dot{\omega}_2=-d_0\phi_2+n_0u_2,\\
x=&\begin{pmatrix}
x_1,x_2,x_3,v_1,v_2,v_3,\phi_1,\omega_1,\phi_2,\omega_2
\end{pmatrix}^\top\in\mathbb{R}^{10},\\
u=&\begin{pmatrix}
u_1,u_2,u_3
\end{pmatrix}^\top\in\mathbb{R}^3,
\end{align*}
 where $(x_1,x_2,x_3)$ are the position, $(v_1,v_2,v_3)$ are the velocities, $(\phi_1,\phi_2)$ denote the pitch and roll, $(\omega_1,\omega_2)$ the pitch and roll rates, and $(u_1,u_2,u_3)$ are the adjustable pitch angle, roll angle and the vertical thrust. 
The parameters are 
$d_0=10,~d_1=8,~n_0=10,~k_T=0.91,~g=9.8,$
and the constraint set is
\begin{align*}
\mathcal{Z}=\{x_1\leq 4,~|\phi_{i}|\leq \pi/4,~|u_{1,2}|\leq \pi/9,~u_3\in[0,2g]\},
\end{align*}
which can be written using $L_{j,x},~L_{j,u}$, compare~\eqref{eq:con_polytope}.
As in~\cite{hu2018real}, we  consider  the problem of stabilizing the steady-state $x_{r}=[3,3,10,0_7]^\top$ with the quadratic stage cost $\ell(x,u)=\|x-x_{r}\|_Q^2$, $Q=\text{diag}[1,1,1,0_7]$ and the initial state $x_0=[0_2,2,0_7]^\top$.

\subsubsection*{Additive disturbances}
We first consider quadratically constrained additive disturbances $w=[w_1,w_2,w_3]^\top$, with $\|w\|^2\leq 1$, as in~\cite{hu2018real}. 
For simplicity, we compute a quadratic incremental Lyapunov function $V_{\delta}(x,z)=\|x-z\|_P^2$, and linear feedback $\kappa(x)=Kx$ offline, such that  $f_{\kappa}(x,v)=f(x,v+\kappa(x))$ is incrementally exponentially stable (Ass.~\ref{ass:contract}, Ass.~\ref{ass:contract_cont}) and the sublevel set $V_{\delta}(x,x_r)\leq 1$ ensures $(x,\kappa(x)+u_r)\in\mathcal{Z}$ and is robust positively invariant for the perturbed dynamics $f_{\kappa}(x,v)+Ew$. 
To this end, we parameterize the Jacobian 
\begin{align*}
\left[\dfrac{\partial f}{\partial x}\right]_{(x,u)}=&A_0+\sum_{i=1}^2\theta_{nl,i}(x) A_i=A_{\theta_{nl}},\quad \dfrac{\partial f}{\partial u}=B,\\
\theta_{i,nl}(x)=&\tan(\phi_i)^2\in\Theta_{i,nl}=[0,\tan^2(\pi/4)]. 
\end{align*}
The desired conditions (Ass.~\ref{ass:contract}) can be ensured by solving the following linear matrix inequalities (LMIs) offline
\begin{align}
\label{eq:LMI_RPI}
\min_{X,Y}&-\log\det(X)\\
\text{s.t. }&
A_{\theta_{nl,i}}X+BY+(A_{\theta_{nl,i}}X+BY)^\top +2\rho_c X\leq 0,\nonumber\\
&\begin{pmatrix}
A_{\theta_{nl,i}}X+BY+(A_{\theta_{nl,i}}X+BY)^\top +\lambda X&E\\
* &-\lambda I_3
\end{pmatrix}\leq 0,\nonumber\\
&~ \theta_{nl,i}\in\text{vert}(\Theta_{1,nl}\times\Theta_{2,nl}),~i=1,\dots,2^2,\nonumber\\
&\begin{pmatrix}
1&L_{j,x}X+L_{j,u}Y\\
*&X
\end{pmatrix}\geq 0, ~j=1,\dots,p,\nonumber
\end{align}
with $P=X^{-1}$, $K=YP$, and $\theta_{nl,i}$ the $2^2$ vertices, compare~\cite{boyd1994linear}\cite[Prop.~5]{koehler2020nonlinear}. 
The constant $\lambda\geq 0$ is due to the application of the S-procedure and can be computed using bi-section. 
The resulting constants (Ass.~\ref{ass:contract}, Prop.~\ref{prop:contract}, \eqref{eq:w_delta_tilde_cont}) are $\lambda=0.1084$, $\rho_c=0.192$, $\delta_{loc}=1$, $c_{\delta,l}=0.016$, 
$c_{\delta,u}=23.63$, $\max_{j}c_j=c_{\max}= 1$. 
In the following, we only consider the pre-stabilized dynamics $f_{\kappa}(x,v)$.  
The continuous-time disturbance bound~\eqref{eq:w_delta_tilde_cont} is given by 
$\overline{w}_c=\max_{\|w\|^2\leq 1}\|E w\|_P= 0.1646$. 

The prediction horizon is set to $T=3~s$ and we consider piece-wise constant input signals $v_{\cdot|t}$ with a sampling time of $h=0.3~s$ ($N=T/h=10$).  
In closed-loop operation, we  apply the input  $u_{\tau}=v^*_{0|t}+\kappa(x_{\tau})$,~$\tau\in[t,t+h]$. 
The discrete-time prediction model is defined with the $4th$ order Runge Kutta method and the step size $h$. 
For simplicity, we only consider the constraints at the sampling instances $\tau=hk$, $k\in\mathbb{N}$. 
The corresponding discrete-time contraction rate is $\rho=e^{-\rho_c h}=0.944$ and the discrete-time disturbances bound~\eqref{eq:w_max_add} is $\overline{w}=\overline{w}_c(1-e^{-\rho_c h})/\rho_c=0.48$. 
As in Section~\ref{sec:additive}, the tube size is given by $s_{\tau|t}=(1-e^{-\rho_c \tau})/\rho_c \overline{w}_c$, with $s_{T|t}=\overline{s}_f=(1-e^{-\rho_c T})/\rho_c\overline{w}_c=\dfrac{1-\rho^N}{1-\rho}\overline{w}=0.37$.

The terminal set is chosen as $\mathcal{X}_f=\{x|~V_{\delta}(x,x_r)\leq \gamma^2\}$, with $\gamma=1/c_{\max}-\overline{s}_f=0.63$ (c.f. Prop.~\ref{prop:LPV_simple}), which satisfies the RPI condition~\eqref{eq:term_dist_RPI} and thus Assumption~\ref{ass:term}. 
The terminal cost $V_f(x)=\|x-x_r\|_{P_f}^2$ is computed based on the linearized model and the Lyapunov equation, such that~\eqref{eq:term_dist_dec} holds with $v=k_f(x)=0$.

We would like to point out that the constants $\rho,~\overline{w}$ can, as done here, be obtained based on continuous-time constants $\rho_c,~\overline{w}_c$. Alternatively, given a fixed sampling time $h$ and a corresponding discrete-time model, the constants can be directly computed. One of the advantages of considering the continuous-time formulation is that the Jacobians $A_{\theta_{nl}},~B$ used in the LMI computation are significantly simpler.

\subsubsection*{Closed-loop simulations}
The corresponding open-loop prediction $x^*_{k|t}$ with the implicitly defined tube $\mathbb{X}_{k|t}=\{x|~V_{\delta}(x,x^*_{k|t})\leq s^2_{k|t}\}$ and an exemplary closed-loop trajectory can be seen in Figure~\ref{fig:quad_add}.

\subsubsection*{Discussion}
The same example\footnote{%
Compared to~\cite{hu2018real}, we introduced a constraint on $\phi_i$, since the system is not well defined for $\phi_{i}=\pm\pi/2$. 
In the nominal case ($w=0$), we were unable to find a feasible solution with the posed constraints and a prediction horizon of $T=1.2~s$ (which is considered in~\cite{hu2018real}). 
}   has been considered in~\cite{hu2018real} with a different robust MPC approach~\cite{villanueva2017robust}, compare Remark~\ref{remark:comput}. 
The crucial difference between the two approaches is that we compute the incremental Lyapunov function $V_{\delta}$ and thus the constraint tightening offline, while the method in~\cite{villanueva2017robust} computes the incremental Lyapunov function $V_{\delta}$ online. 
Although in this example we can consider an equal magnitude of disturbances, in general the approach presented in this paper may be more conservative compared to the approach in~\cite{villanueva2017robust}. 
This conservatism is (partially) due to the difference between the offline chosen incremental Lyapunov function $V_{\delta}$ and the online optimized one in~\cite{villanueva2017robust}.
On the other hand, in case  of constantly bounded additive disturbances, the online computational demand of the proposed approach is equivalent to nominal MPC, while the online computation in~\cite{villanueva2017robust,hu2018real} is increased by a factor of 500 compared to the nominal implementation, compare\footnote{%
In~\cite{hu2018real} it is reported that the average computational time per real time iteration (RTI) with the nominal certainty-equivalent MPC is $159~\mu s$, and $81.6~ms$ with the tube MPC, thus increasing the computational demand by $81.6ms/159\mu s\approx 500$.
}~\cite{hu2018real}. 
This is due to the fact that the proposed approach requires \textit{no} additional input or state variable in case of constantly bounded additive disturbances (compare Section~\ref{sec:additive}), while in~\cite{hu2018real}, $(n^2+n)/2=55$ additional state variables and $n\cdot m=30$ additional decision/input variables are needed.

If a Lipschitz-based approach is considered ($P_{\delta}=I_n$, $K=0$, compare~\cite{marruedo2002input}), we have Lipschitz constant $L=1.8$, $\overline{w}=0.15$ and thus for $N\geq 3$ the predicted tube is larger than the constraint set or, alternatively, for $N=10$ the magnitude of the disturbances has to be reduced by $99\%$ to ensure that the tube $s$ is contained in the constraint set $\mathcal{Z}$ (independent of the terminal region). 

Thus, in the considered scenario the proposed approach has a computational complexity which is equivalent to~\cite{marruedo2002input} and significantly reduced compared to~\cite{villanueva2017robust}, while the conservatism is similar to~\cite{villanueva2017robust} and significantly reduced compared to~\cite{marruedo2002input}.

\subsubsection*{Parametric uncertainty}
In the following, the parameters $n_0,~d_0$ are subject to a (possibly time-varying) uncertainty of $\pm 10\%$ and we neglect the  additive disturbances $(w=0)$.  
The uncertain nonlinear system dynamics can be written as 
\begin{align*}
\dot{x}=&f(x,u)+A_{\theta_{par}}x+B_{\theta_{par}}u,\\
A_{\theta_{par}}=&\sum_{i=1}^2\theta_{par,i}A_{par,i},~B_{\theta_{par}}=\sum_{i=1}^2 B_{par,i},~ \theta_{par}\in[-1,1]^2. 
\end{align*}
Again, we compute a simple quadratic incremental Lyapunov function $V_{\delta}(x,z)=\|x-z\|_P^2$ and a linear feedback $\kappa(x)=Kx$ offline. 
Similar to the discrete-time LPV design in Prop.~\ref{prop:LPV_simple}, the function\footnote{%
For the numerical implementation, we use $\|x\|_P\approx\sqrt{x^\top Px+\epsilon}$, with $\epsilon=10^{-4}$ and implement $\max_i$ using $4$ inequality constraints. }    
\begin{align}
\label{eq:w_tilde_quadrotor}
\tilde{w}_{\delta}(x,v,s)=&\max_i \|A_{\theta_{par,i}}x+B_{\theta_{par,i}}(v+Kx)\|_P+L_{w}s,\nonumber\\
L_{w}=&\max_i \|P^{1/2}(A_{\theta_{par,i}}+B_{\theta_{par,i}}K)P^{-1/2}\|,
\end{align}
satisfies the continuous-time condition~\eqref{eq:w_delta_tilde_cont} with the vertices $\theta_{par,i}\in\{-1,1\}^2$, compare Remark~\ref{rk:lpv_QP}. 
Note that due to the symmetry in $\theta_{par}$, the function $\tilde{w}_{\delta}$ in~\eqref{eq:w_tilde_quadrotor} can be implemented using only $2$ inequality constraints with $\theta_{par,i}\in\{(1,1);(1,-1)\}$. 
For the considered approach $\rho_c>L_{w}$ is crucial (Prop.~\ref{prop:LPV_simple}, Remark~\ref{remark:L_theta}). 
Thus, we compute the matrices $P,K$ similar to~\eqref{eq:LMI_RPI} by enforcing $\rho_c\geq 0.4$ and $L_{w}\leq 0.3$, with the following additional constraint
\begin{align*}
\begin{pmatrix}
L_{w}X&A_{\theta_{par}}X+B_{\theta_{par,i}}Y\\
(A_{\theta_{par}}X+B_{\theta_{par,i}}Y)^\top &L_{w}X 
\end{pmatrix}\geq 0. 
\end{align*}
We choose the terminal set
 $\mathcal{X}_f=\{(x,s)\in\mathbb{R}^{n+1}|~s\in[0,\overline{s}_f],~V_{\delta}(x,x_r)\leq \gamma^2\}$, with $\overline{s}_f=\frac{L_{w}}{c_{\max}\rho_c}=0.75$, $\gamma=\frac{\rho_c-L_{w}}{c_{\max}\rho_c}=0.25$ (compare Prop.~\ref{prop:term}--\ref{prop:LPV_simple}).

\begin{remark}
\label{rk:example_cont_time} (Continuous-time formulations)
The simplicitiy of the function~\eqref{eq:w_tilde_quadrotor} shows the benefit of considering the \textit{continuous-time} model.
In particular, as detailed in Prop.~\ref{prop:LPV_simple}, we can consider the same kind of function in a discrete-time formulation, if we have a discrete-time model affine in the parameters. 
However, while the continuous-time system is affine in the uncertain parameters ($d_0,n_0$), the non-trivial discretization (in this case 4th order Runge Kutte) results in a discrete-time model which is nonlinear in the uncertain parameters.
Thus, in case that a physics-based continuous-time model is considered, it may be easier to also formulate the design of $\tilde{w}_{\delta}$ in continuous-time (App.~\ref{app:cont}). 
This is also the reason why research in robust adaptive MPC (dealing with parametric uncertainty) is typically formulated in continuous-time, compare e.g.~\cite{adetola2011robust}.

The only drawback of continuous-time formulations is the needed discretization for the continuous-time tube predictions in~\eqref{eq:MPC_cont_s}, \eqref{eq:MPC_cont_w}, compare Remark~\ref{rk:implement} in Appendix~\ref{app:cont} for details on possible implementations. 
To allow for a simple implementation, we assume that $(x,u)$ are constant in the sampling interval $h$ (which is clearly an approximation). 
Thus, the (piece-wise linear) continuous-time dynamic $\dot{s}=(L_w-\rho_c)s+\tilde{w}_{\delta}(x,u,0)$ in~\eqref{eq:MPC_cont_s} can be exactly discretized resulting in 
\begin{align*}
s^+=e^{-(\rho_c-L_w)h}s+\frac{1-e^{-(\rho_c-L_w)h}}{\rho_c-L_w}\tilde{w}_{\delta}(x,u,0),
\end{align*}
 which has been used in the implementation. 
In the considered example, this formula is almost equivalent to a simple Euler discretization of the tube dynamics with $s^+=(1-h\rho_c )s+h\tilde{w}_{\delta}(x,u,s)$. 
\end{remark}

Due to the significant increased uncertainty in the pitch/roll dynamics, we now require a larger prediction horizon of $T=4.5~s$ ($N=15$) to find a safe trajectory. 
The (continuous-time) disturbance bound is given by $\overline{w}_c=\gamma\rho_ce^{\rho_c T}= 0.6$, which satisfies the conditions on the terminal region.  
The resulting predicted trajectory can be seen in Figure~\ref{fig:quad_add}.  
The constraints $w_{\tau|t}\leq \overline{w}$ and $s_{T|t}\leq \overline{s}_f$ implicitly constrain the optimized trajectory to avoid maneuvers which are too risky (in terms of the uncertainty), while the constraint tightening enforces a safe distance to critical constraints commensurate with the predicted uncertainty $s_{\cdot|t}$. 
The combination of these features ensures that the MPC acts cautiously and ensures constraint satisfaction despite uncertainty. 
In particular, the peak values of $\phi_1,\phi_2,u_1,u_2$ corresponding to the parametric uncertainty are reduced by $30-70\%$ compared to the first scenario without parametric uncertainty. Alternatively, if one would like to use a constant bound (Sec.~\ref{sec:additive}) with the considered terminal region (instead of considering the state and input dependency of the uncertainty), this would require $w_{\tau|t}\leq 0.36$, which is ensured if either the maximal allowed values of $\phi_1,\phi_2,u_1,u_2$ or the parameter uncertainty in $d_0,n_0$ is reduced by $85\%$. Thus, the more detailed state and input dependent characterization of the uncertainty significantly reduces the conservatism of the robust MPC. 

In this example, the proposed nonlinear robust MPC scheme requires roughly $1/m\approx33\%$ more decision variables and $2\cdot N$ additional nonlinear inequality constraints~\eqref{eq:w_tilde_quadrotor}, resulting in roughly $4.5$ times as large computation time compared to a nominal MPC scheme. 
By replacing~\eqref{eq:w_tilde_quadrotor} with the simpler and more conservative formula
\begin{align*}
\tilde{w}_{\delta}(x,v,s):=&\|A_{\theta_{par,1}}x\|_P+\|B_{\theta_{par,1}}(v+Kx)\|_P+L_{w,2}s,\\
L_{w,2}:=& \|P^{1/2}A_{\theta_{par,1}}P^{-1/2}\|+\|P^{1/2}B_{\theta_{par,1}}KP^{-1/2}\|,
\end{align*}
we can reduce the online computational demand by $32\%$, at the cost of increased conservatism, compare Remark~\ref{rk:lpv_QP}. 

Similar to the first scenario, Lipschitz based approaches~\cite{pin2009robust} ($P_{\delta}=I_n,~K_\delta=0$) cannot be applied since the propagated tube size  becomes overly large over the prediction horizon with $L_d^N\geq 6\cdot10^3$. 

\begin{figure}[hbtp]
\begin{center}
\includegraphics[width=0.5\textwidth]{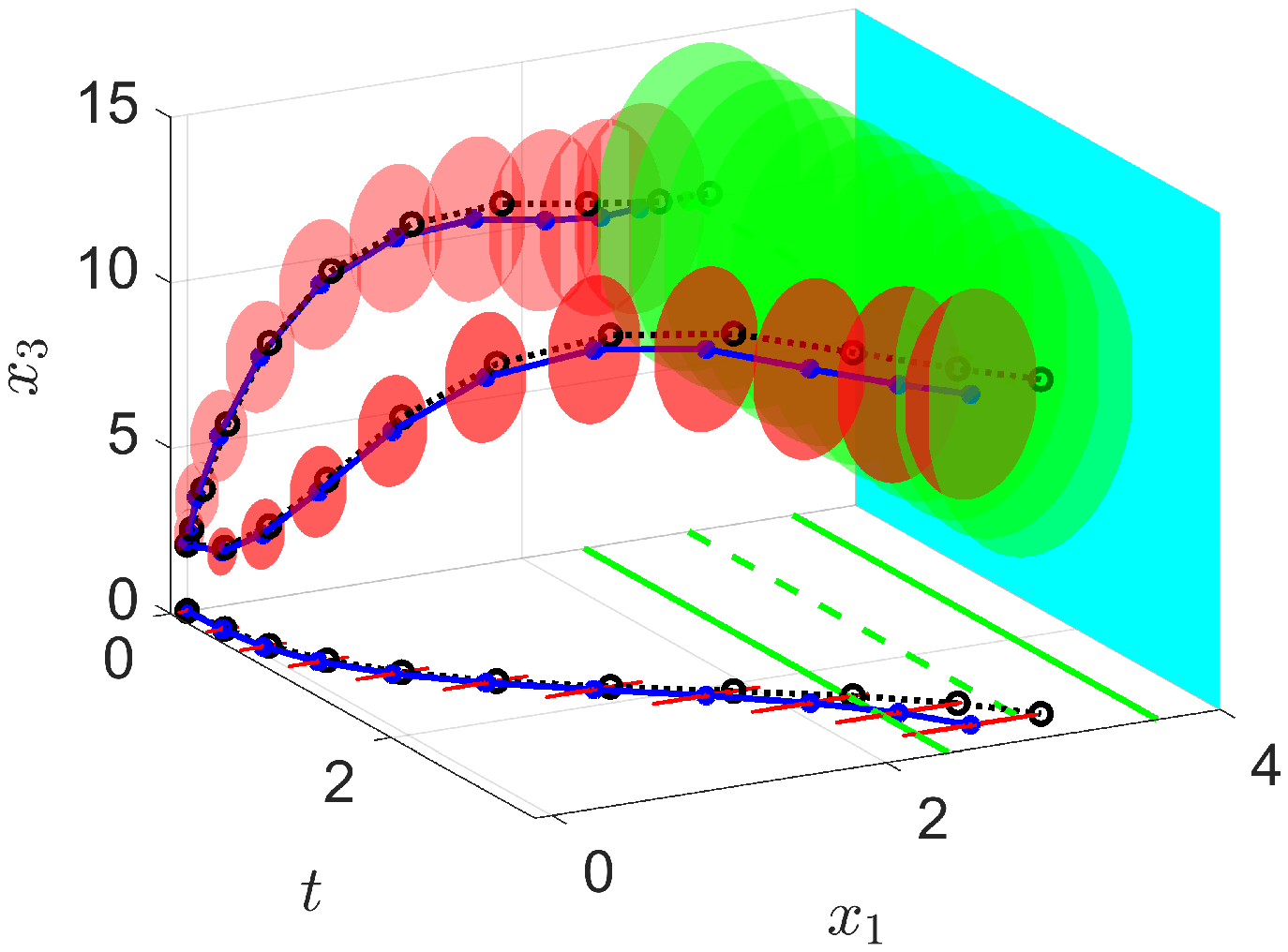}\\
\includegraphics[width=0.5\textwidth]{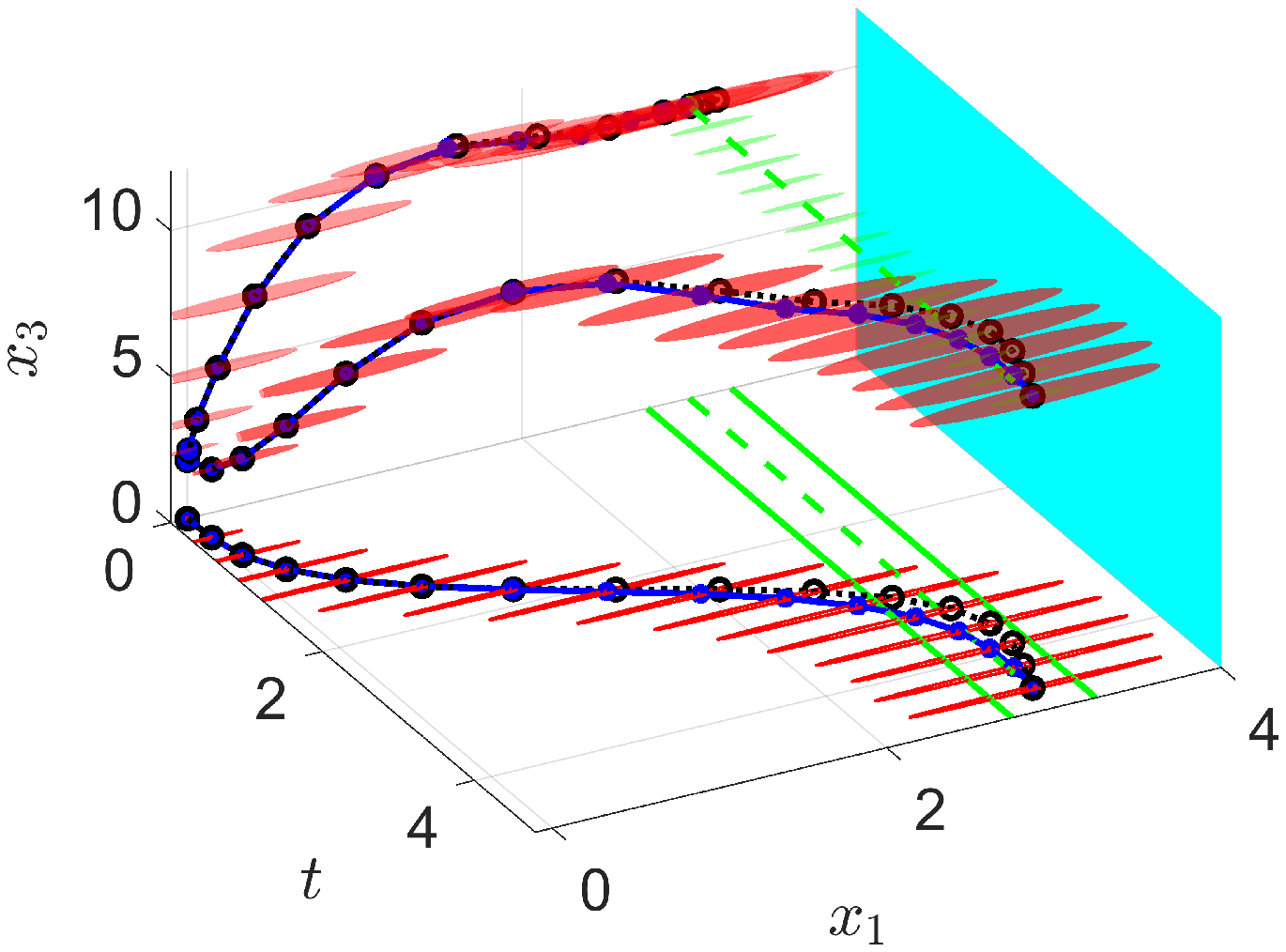}
\end{center}
\caption{Top: Additive disturbances. Bottom: Parametric uncertainty. Initially predicted trajectory $x^*_{\cdot|0}$ (blue, solid), uncertainty tube $\mathbb{X}^*_{\cdot|0}$ (red ellipse), terminal set $\mathcal{X}_f$ (green ellipse), state constraint ($x_1\leq 4$, turquoise plane) and an exemplary closed-loop trajectory $x_t$ (black, dotted). The projections on the $t$--$x_1$ and $x_1$--$x_3$ plane are also shown for illustration. }
\label{fig:quad_add}
\end{figure}
To summarize, in the considered numerical example we have demonstrated: (i) the computational efficiency relative to~\cite{villanueva2017robust,hu2018real}, (ii)  significantly reduced conservatism compared to~\cite{marruedo2002input,pin2009robust} and comparable conservatism to~\cite{villanueva2017robust,hu2018real}, (iii) reduced conservatism using a state and input dependent uncertainty characterization, (iv) the possibility to trade off conservatism and computational complexity in the design of $\tilde{w}_\delta$. 

\section{Conclusion}
\label{sec:sum}
We have presented a nonlinear robust tube MPC framework based on incremental stabilizability that ensures robust constraint satisfaction and recursive feasibility despite disturbances and uncertainty. 
The scheme is applicable to \textit{nonlinear systems}, can incorporate \textit{general state and input dependent uncertainty descriptions} and is \textit{easy} to implement. 
Furthermore, the framework allows for an intuitive trade-off between computational demand and conservatism by using more detailed, and hence more complex, descriptions of the uncertainty $\tilde{w}_{\delta}$.  
We have demonstrated the applicability of the proposed  framework in comparison to state of the art robust MPC approaches with a nonlinear benchmark example. 

Recent extensions of the proposed robust MPC framework to adaptive MPC  and stochastic MPC with chance constraints  can be found in~\cite{kohler2019robust} and \cite{schluter2020constraint}, respectively. 
Current research focuses on extending this framework to output feedback~\cite{Kohler2019Output}.

\bibliographystyle{ieeetran}  
\bibliography{Literature}  
%
\clearpage 
\appendix
\label{sec:app}
In Appendix~\ref{app:cont}, the main results of the paper are extended to continuous-time systems. 
In Appendix~\ref{app:nonlin_constraints}, the theory is extended to general (continuous) nonlinear constraint. 
%
\subsection{Continuous-time systems}
\label{app:cont}
In the following, we detail the exact conditions, functions and properties for the continuous-time formulation of the proposed nonlinear robust MPC framework.
\subsubsection*{Setup}
We consider a nonlinear perturbed continuous-time system, which is given by the following ordinary differential equation (ODE) 
\begin{align*}
\dfrac{d}{dt} x=\dot{x}=f_w(x,u,d)=f(x,u)+d_w(x,u,d),
\end{align*}
 with state $x\in\mathbb{R}^n$, control input $u\in\mathbb{R}^m$, disturbance $d\in\mathbb{D}\subset\mathbb{R}^q$,  perturbed system $f_w$, nominal model $f$, model mismatch $d_w$ and some initial condition $x_0$. 
We impose state and input constraints  
\begin{align}
\label{eq:constraint_con}
(x_t,u_t)\in \mathcal{Z},\quad t\in\mathbb{R}_{\geq 0}, 
\end{align}
with some compact nonlinear constraint set 
\begin{align*}
\mathcal{Z}=\{(x,u)\in\mathbb{R}^{n+m}|~g_j(x,u)\leq 0,~j=1,\dots,p\}\subset\mathbb{R}^{n+m}.
\end{align*}
Define the projected state constraint set
\begin{align*}
\mathcal{Z}_x=\{x\in\mathbb{R}^n|~\exists u\in\mathbb{R}^m,~(x,u)\in\mathcal{Z}\}. 
\end{align*}
\begin{remark}
For the theoretical analysis, we consider the case where the constraints~\eqref{eq:constraint_con} need to be satisfied for all continuous point in time $t\in \mathbb{R}_{\geq 0}$. 
In a standard implementation, the constraints are often only enforced at sampling time points $t_k=kh$,~$k\in\mathbb{N}$ with the sampling time $h>0$. 
Continuous-time ($t\in\mathbb{R}$) constraint satisfaction can be achieved by enforcing the constraints at sampling points $t_k$ in combination with an appropriate constraint tightening, compare~\cite{fontes2019guaranteed,magni2004model}. 
\end{remark}
\begin{assumption}
\label{ass:hat_w_cont}
For each $(x,u)\in\mathcal{Z}$, there exists a compact set $\mathcal{W}(x,u)\subset\mathbb{R}^n$, such that the model mismatch $d_w$ satisfies $d_w(x,u,d)\in\mathcal{W}(x,u)$ for all $d\in\mathbb{D}$. 
\end{assumption} 
 Such a description includes additive disturbances, multiplicative disturbances, more general nonlinear disturbances, and/or unmodeled nonlinearities. 
\subsubsection*{Local Incremental stabilizability}
In order to provide theoretical guarantees for robust stabilization, we assume that the system is locally incrementally stabilizable. 
\begin{assumption}
\label{ass:contract_cont}
There exist a control law $\kappa:\mathbb{R}^n\rightarrow\mathbb{R}^m$, an incremental Lyapunov function $V_{\delta}:\mathbb{R}^n\times\mathbb{R}^n\rightarrow\mathbb{R}_{\geq 0}$, which is continuous in the first argument and satisfies $V_{\delta}(z,z)=0$ for all $z\in\mathbb{R}^n$, and parameters $c_{\delta,l}$, $c_{\delta,u}$, $\delta_{\text{loc}}$, $\kappa_{\max}$, $\rho_c>0$, such that the following properties hold for all $(z,v+\kappa(z))\in\mathcal{Z}$ with $V_{\delta}(x,z)\leq \delta_{\text{loc}}$: 
\begin{subequations}
\begin{align}
c_{\delta,l}\|x-z\|^2\leq V_{\delta}(x,z)\leq& c_{\delta,u}\|x-z\|^2,\\
\|\kappa(x)-\kappa(z)\|^2\leq& \kappa_{\max}{{V_{\delta}(x,z)}},\\
\label{eq:contract_cont}
\dfrac{d}{dt}V_{\delta}(x,z)=&\dfrac{\partial V_{\delta}}{\partial x}|_{(x,z)}\dot{x}+\dfrac{\partial V_{\delta}}{\partial z}|_{(x,z)}\dot{z}\nonumber\\
\leq& -2\rho_c  V_{\delta}(x,z),
\end{align}
\end{subequations}
 with $\dot{x}=f(x,\kappa(x)+v)$, $\dot{z}=f(z,\kappa(z)+v)$.  
\end{assumption}
The dynamics and constraints can be reformulated with $f_{\kappa}(x,v)=f(x,v+\kappa(x))$, $g_{j,\kappa}(x,v)=g_j(x,v+\kappa(x))$.  
Thus, Assumption~\ref{ass:contract_cont} ensures that the system~$f_\kappa$ is exponentially incrementally stable, similar to~\cite{bayer2013discrete}. 
Note that this assumption is stronger than Assumption~\ref{ass:contract} due to the parametrization of $\kappa$, which will be essential to allow for a simple finite parametrization of the decision variable $v_{\cdot|t}$ in the MPC optimization problem~\eqref{eq:MPC_cont_kappa}. 
In the following, we consider the continuous-time input $u=v+\kappa(x)$ with a piece-wise constant input $v$ optimized online. 

 \subsubsection*{Efficient disturbance description}
\begin{assumption}
\label{ass:approx_V_cont}
There exists a function $\tilde{w}_{\delta}:\mathcal{Z}\times\mathbb{R}_{\geq 0}\rightarrow\mathbb{R}_{\geq 0}$, such that for any $(x,z,v+\kappa(z))\in\mathbb{R}^n\times\mathcal{Z}$ with $V_{\delta}(x,z)\leq c^2\leq \delta_{loc}$ and any $d_w\in\mathcal{W}(x,v+\kappa(x))$, we have
\begin{align}
\label{eq:w_delta_tilde_cont}
\dfrac{d}{dt} \sqrt{V_{\delta}(x,z)}\leq -\rho_c \sqrt{V_{\delta}(x,z)}+\tilde{w}_{\delta}(z,v,c),\\
\dot{x}=f_{\kappa}(x,v)+d_w,~\dot{z}=f_{\kappa}(z,v).  \nonumber
\end{align}
Furthermore, $\tilde{w}_{\delta}$ satisfies the monotonicity property in~\eqref{eq:w_delta_tilde}. 
\end{assumption}
\begin{proposition}
Let Assumption~\ref{ass:contract_cont} hold. 
Suppose there exists a function $\tilde{w}_{\delta}$ that satisfies the monotonicity property~\eqref{eq:w_delta_tilde} and the following condition for any $(x,z,v+\kappa(z))\in\mathbb{R}^n\times\mathcal{Z}$, $d_w\in\mathcal{W}(x,v+\kappa(x))$:
\begin{align}
\label{eq:w_delta_tilde_intermediate}
\dfrac{\partial V_{\delta}}{\partial x}|_{(x,z)}d_w\leq& 2\sqrt{V_{\delta}(x,z)}\tilde{w}_{\delta}(z,v,c).
\end{align}
Then this function $\tilde{w}_{\delta}$ satisfies Assumption~\ref{ass:approx_V_cont}. 
\end{proposition}
\begin{proof}
Using $\dfrac{d}{dt}\sqrt{V_{\delta}}=\dfrac{1}{2\sqrt{V_{\delta}}}\dfrac{d}{dt}V_{\delta}$, the condition in Assumption~\ref{ass:approx_V_cont} follows from
\begin{align*}
\dfrac{d}{dt}V_{\delta}(x,z)=&\dfrac{\partial V_{\delta}}{\partial x}|_{(x,z)}(f_{\kappa}(x,v)+d_w) +\dfrac{\partial V_{\delta}}{\partial z}|_{(x,z)}f_{\kappa}(z,v)\\
\stackrel{\eqref{eq:contract_cont},\eqref{eq:w_delta_tilde_intermediate}}{\leq}& -2\rho_c V_{\delta}(x,z)+2\tilde{w}_{\delta}(z,v,c)\sqrt{V_{\delta}(x,z)}.
\end{align*}
\end{proof}
We consider the problem of stabilizing the origin and assume that $(0,\kappa(0))\in\text{int}(\mathcal{Z})$, $f_{\kappa}(0,0)=0$. 
The open-loop cost over a prediction horizon $T\in\mathbb{R}$ of a predicted state and input sequence $x_{\cdot|t},~v_{\cdot|t}$ is defined as
\begin{align*}
&J_T(x_{\cdot | t},v_{\cdot|t})=\int_{\tau=0}^{T}\ell(x_{\tau|t},v_{\tau|t})+V_f(x_{T|t}),
\end{align*}
with some positive definite stage cost $\ell$ and terminal cost $V_f$. 

The MPC optimization problem is given by
\begin{subequations}
\label{eq:MPC_cont_kappa}
\begin{align}
&\min_{v_{\cdot|t},w_{\cdot|t},x_{\cdot|t},s_{\cdot|t}}J_T(x_{\cdot | t},v_{\cdot|t})\\ 
\label{eq:MPC_cont_init}
\text{s.t. }&x_{0|t}=x_t,~s_{0|t}=0,\\
\label{eq:MPC_cont_dyn}
&\dot{x}_{\tau|t}=f_{\kappa}(x_{\tau|t},v_{\tau|t}),\\
\label{eq:MPC_cont_s}
&\dot{s}_{\tau|t}=-\rho_cs_{\tau|t}+w_{\tau|t},\\
\label{eq:MPC_cont_w}
&w_{\tau|t}\geq \tilde{w}_{\delta}(x_{\tau|t},v_{\tau|t},s_{\tau|t}),\\
\label{eq:real_constraints_cont}
&g_{j,\kappa}(x_{\tau|t},v_{\tau|t})+c_js_{\tau|t}\leq 0,\\
\label{eq:tube_bound_MPC_cont}
&s_{\tau|t}\leq \overline{s},~
 {w}_{\tau|t}\leq \overline{w},\\
\label{eq:terminal_region_cont}
 &(x_{T|t},s_{T|t}) \in \mathcal{X}_f,\\
& \tau\in[0,T],\quad j=1,\dots, p,\nonumber
\end{align}
\end{subequations}
We consider a piece-wise constant parametrization of $v$ with a sampling time $h$ and implement the continuous-time input 
\begin{align}
\label{eq:close_cont}
u_{\tau+t}=v^*_{\tau|t}+\kappa(x_{\tau+t}),~\tau\in[0,h].
\end{align}
The conditions on the stage cost $\ell$ and the constraint set $\mathcal{Z}$ are equivalent to the discrete-time case with 
Assumptions~\ref{ass:stage_cost}--\ref{ass:gen_nonlin_con_Lipschitz}, with exponential summability replaced by exponential integrability, i.e. $\int_0^{\infty} \alpha_c(r e^{-\rho_c \tau})d\tau=:\alpha_{c,\rho_c}(r)\in\mathcal{K}_{\infty}$.
Correspondingly, Proposition~\ref{prop:contract} is also equivalent. 

\subsubsection*{Terminal ingredients}
The assumption regarding the terminal ingredients (Ass.~\ref{ass:term_2}) changes as follows.
\begin{assumption}
\label{ass:term_cont}
There exist a piece-wise constant feedback $k_f:\mathbb{R}^n\rightarrow\mathbb{R}^m$, a terminal cost function $V_f:\mathbb{R}^n\rightarrow \mathbb{R}_{\geq 0}$, a terminal set $\mathcal{X}_f\subset\mathbb{R}^{n+1}$, a constant $\overline{w} \in \mathbb{R}_{\geq 0}$ and a function $\alpha_f\in\mathcal{K}$, such that for any $(x_f,s_f)\in\mathcal{X}_f$, any $w\in\dfrac{1-e^{-\rho_c h}}{\rho_c}[\overline{w}_{\min},\overline{w}]$, and any trajectories $x_{\tau},s_{\tau},~\tau\in[0,h]$ satisfying
\begin{align*}
 \dot{s}_\tau\leq& -\rho_c s_\tau+\tilde{w}_{\delta}(x_\tau,k_f(x_f), s_\tau), \\
\dot{x}_\tau=&f_{\kappa}(x_\tau,k_f(x_f)),
\end{align*}
with initial conditions
\begin{align*}
s_0\leq& s_f-e^{-\rho_c (T-h)}w,\\
V_{\delta}(x_0,x_f)\leq& e^{-2\rho_c (T-h)}w^2,
\end{align*}
the following properties hold for any $\tau\in[0,h]$:
\begin{subequations}
\label{eq:term_dist2_cont}
\begin{align}
\label{eq:term_dist_dec2_cont}
V_f(x_h)- V_f(x_f)\leq&-\int_0^h\ell(x_{\tau},k_f(x_f))d\tau+\alpha_f(w),\\
\label{eq:term_dist_RPI2_cont}
(x_h,s_h)\in&\mathcal{X}_f, \\
\label{eq:w_f_cont}
\tilde{w}_{\delta}(x_\tau,k_f(x_f), s_\tau)\leq & \overline{w},\\ 
\label{eq:term_dist_con2_cont}
g_{j,\kappa}(x_\tau,k_f(x_f))+c_j {s}_\tau\leq & 0,\quad j=1,\dots p,\\
\label{eq:term_s_cont}
s_\tau\leq& \overline{s}.
\end{align}
\end{subequations}
\end{assumption}
Note that the requirement~\eqref{eq:term_dist_con2_cont} is stricter than assuming constraint satisfaction in the terminal set, since $(x_{\tau},s_{\tau})\in\mathcal{X}_f$, $\tau\in(0,h)$ does not necessarily hold. 
Furthermore, we restrict ourselves to piece-wise constant terminal controllers to allow for a simple implementation. 

The following theorem establishes the closed-loop properties of the proposed nonlinear robust MPC scheme. 
\begin{theorem}
\label{thm:cont} 
Let Assumptions~\ref{ass:stage_cost}--\ref{ass:gen_nonlin_con_Lipschitz}, \ref{ass:hat_w_cont}--\ref{ass:term_cont} hold, and suppose that~\eqref{eq:MPC_cont_kappa} is feasible at $t=0$. 
Then~\eqref{eq:MPC_cont_kappa} is recursively feasible, the constraints~\eqref{eq:constraint_con} are satisfied and the origin is practically asymptotically stable for the resulting closed-loop system~\eqref{eq:close_cont}.
\end{theorem}
\begin{proof}
The proof is structured analogous to the  proof in Theorem~\ref{thm:main}.\\
\textbf{Part I.} Candidate solution: 
For convenience, define $v^*_{\tau|t}=k_f(x^*_{T|t})$, $w^*_{\tau|t}=\tilde{w}_{\delta}(x_{\tau|t},v_{\tau|t},s_{\tau|t})$, $\tau\in[T,T+h]$ with $k_f$ from Ass.~\ref{ass:term_cont} and $x^*_{\tau|t}$, $s^*_{\tau|t}$, $\tau\in[T,T+h]$ according to~\eqref{eq:MPC_cont_dyn}--\eqref{eq:MPC_cont_s}. 
Consider the candidate solution
\begin{subequations}
\label{eq:cand_cont}
\begin{align}
 x_{0|t+h}=&x_{t+h},\quad s_{0|t+h}=0,\\
v_{\tau|t+h}=&v^*_{\tau+h|t},~\tau\in[0,T],\\
\label{eq:cand_cont_w}
w_{\tau|t+h}=&\tilde{w}_{\delta}(x_{\tau|t+h},u_{\tau|t+h},s_{k|t+h}),~,~\tau\in[0,T],
\end{align}
\end{subequations}
with $x_{\cdot|t+h},~s_{\cdot|t+h}$ according to~\eqref{eq:MPC_cont_dyn}--\eqref{eq:MPC_cont_s}. 
Assumption~\ref{ass:approx_V_cont}, Inequality~\eqref{eq:w_delta_tilde_cont} and \eqref{eq:MPC_cont_s} ensure
\begin{align*}
&\sqrt{V_{\delta}(x_{t+h},x^*_{h|t})}\leq s^*_{h|t}{\leq} \overline{s}{\leq}\sqrt{\delta}_{loc}. 
\end{align*}
Using the contractivity in Assumption~\ref{ass:contract_cont} recursively, we get  
\begin{align}
\label{eq:prop_robust_1_2_cont}
&\sqrt{V_{\delta}(x_{\tau|t+h},x^*_{\tau+h|t})}
\leq e^{-\rho_c\tau}s^*_{h|t}\leq \sqrt{\delta_{loc}},
\end{align}
for any $\tau\in[0,T]$. \\
\textbf{Part II.} Tube dynamics: In the following, we show by induction that the following inequalities hold for $\tau\in[0,T]$:
\begin{align}
\label{eq:tube_shrinking_cont}
s_{\tau|t+h}\leq & {s}^*_{\tau+h|t}-e^{-\rho_c \tau}{s}^*_{h|t},\\
\label{eq:disturbance_shrinking_cont}
{w}_{\tau|t+h}\leq & {w}^*_{\tau+h|t}.
\end{align}
For $\tau=0$, Inequality~\eqref{eq:tube_shrinking_cont} is trivially satisfied with equality. 
Furthermore, Assumption~\ref{ass:approx_V_cont}, \eqref{eq:w_delta_tilde} ensures $w_{0|t+h}\stackrel{\eqref{eq:w_delta_tilde},\eqref{eq:cand_cont_w}}{\leq} \tilde{w}_{\delta}(x^*_{h|t},v^*_{h|t},s^*_{h|t})\stackrel{\eqref{eq:MPC_cont_w}}{\leq} w^*_{h|t}$, with $c_2 = s_{0|t+h} = 0$, $c_1 = {s}^*_{h|t}$,  
\begin{align*}
&V_{\delta}(x_{0|t+h},x^*_{h|t}) \leq   [s^*_{h|t}]^2  = (c_1-c_2)^2.
\end{align*}
Suppose~\eqref{eq:tube_shrinking_cont} holds at some $\tau$. 
Consider $c_2 = s_{\tau|t+h}$, $c_1 = {s}^*_{\tau+h|t}$. The incremental Lyapunov function satisfies
\begin{align*}
\sqrt{V_{\delta}(x_{\tau|t+h},x^*_{\tau+h|t})}\stackrel{\eqref{eq:prop_robust_1_2_cont}}{\leq} e^{-\rho_c \tau} s^*_{h|t} \stackrel{\eqref{eq:tube_shrinking_cont} }{\leq} c_1-c_2.
\end{align*}
Thus, Assumption \ref{ass:approx_V_cont}, \eqref{eq:w_delta_tilde}, and \eqref{eq:MPC_cont_w} imply ${w}_{\tau|t+h}\leq {w}^*_{\tau+h|t}$.
Consider the differentiation
\begin{align*}
&\dfrac{d}{d\tau}\left[ s_{\tau|t+h}-s^*_{\tau+h|t}+e^{-\rho_c\tau}s^*_{h|t}\right]\\
\stackrel{\eqref{eq:MPC_cont_s}}{=}&\rho_c(s^*_{\tau+h|t}-s_{\tau|t+h})+w_{\tau|t+h}-w^*_{\tau+h|t}\\
&-\rho_ce^{-\rho_c \tau}s^*_{h|t}\\
\leq &-\rho_c(s_{\tau|t+h}-s^*_{\tau+h|t}+e^{-\rho_c \tau}s^*_{h|t}).
\end{align*}
Given that~\eqref{eq:tube_shrinking_cont} holds for $\tau=0$, the Inequality~\eqref{eq:tube_shrinking_cont} holds for all $\tau\in[0,T]$.  \\
\textbf{Part III.} State and input constraint satisfaction: \\
For $\tau\in[0,T-\delta]$, we have  
\begin{align*}
&g_{j,\kappa}(x_{\tau|t+h},v_{\tau|t+h})+c_js_{\tau|t+h}\\
\stackrel{\eqref{eq:lipschitz}}\leq& g_{j,\kappa}(x^*_{\tau+h|t},v^*_{\tau+h|t})\\
&+ e^{-\rho_c \tau}c_j{s}^*_{h|t}+c_js_{\tau|t+h}\\
\stackrel{\eqref{eq:tube_shrinking_cont}}{\leq}&g_{j,\kappa}(x^*_{\tau+h|t},v^*_{\tau+h|t})
+c_js^*_{\tau+h|t}
\stackrel{\eqref{eq:real_constraints_cont}}\leq 0.
\end{align*}
The terminal condition~\eqref{eq:terminal_region_cont} ensures constraint satisfaction for $\tau\in[T-h,T]$ with  
\begin{align*}
&g_{j,\kappa}(x_{\tau|t+h},v_{\tau|t+h})+c_js_{\tau|t+h}\\
&\stackrel{\eqref{eq:lipschitz}\eqref{eq:prop_robust_1_2_cont}\eqref{eq:tube_shrinking_cont}}\leq g_{j,\kappa}(x^*_{\tau+h|t},v^*_{\tau+h|t})+c_js^*_{\tau+h|t}
\stackrel{\eqref{eq:term_dist_con2_cont}}{\leq}  0.  
\end{align*}
\textbf{Part IV.} Tube bounds \eqref{eq:tube_bound_MPC_cont}: 
Inequalities~\eqref{eq:tube_shrinking_cont} and \eqref{eq:disturbance_shrinking_cont} ensure that \eqref{eq:tube_bound_MPC_cont} holds for $\tau\in[0,T-h]$.
For $\tau\in[T-h,T]$,  $(x^*_{T|t},s^*_{T|t})\in\mathcal{X}_f$ implies
\begin{align}
&s_{\tau|t+h}\stackrel{\eqref{eq:tube_shrinking_cont}}{\leq}  s^*_{\tau+h|t}\stackrel{\eqref{eq:term_s_cont}}{\leq} \overline{s},\nonumber\\
&{w}_{\tau|t+h} \stackrel{\eqref{eq:disturbance_shrinking_cont}} \leq {w}^*_{\tau+h|t}\nonumber\\
=& w_\delta(x^*_{\tau+h|t}, v^*_{\tau+h|t}, s^*_{\tau+h|t})  
\label{eq:bound_w_N_cont}
\stackrel{\eqref{eq:w_f_cont}}\leq     \overline{w}. 
\end{align}
\textbf{Part V.} Terminal constraint satisfaction \eqref{eq:terminal_region_cont}: 
The terminal state and terminal tube size satisfy
\begin{align*}
&\sqrt{V_{\delta}(x_{T-h|t+h},x^*_{T|t})}
 \stackrel{\eqref{eq:prop_robust_1_2_cont}}{\leq} e^{-\rho_c(T-h)} {s}^*_{h|t},\\
&s_{T-h|t+h}\leq s^*_{T|t}-e^{-\rho_c (T-h)}s^*_{h|t}\\
&\dfrac{d}{d\tau} s_{\tau|t+h}\leq -\rho_cs_{\tau|t+h}+w_{\tau|t+h},\quad \tau\in[T-h,T].
\end{align*}
Using $w_{\tau|t}\in[\overline{w}_{\min},\overline{w}]$ and the linear dynamics in $s$~\eqref{eq:MPC_cont_s}, we have $s^*_{h|t}\in(1-e^{-\rho_c h})/\rho_c\cdot[\overline{w}_{\min},\overline{w}]$. 
Thus Ass.~\ref{ass:term_cont} and~\eqref{eq:term_dist_RPI2_cont} ensure $(x_{T|t+h},s_{T|t+h})\in\mathcal{X}_f$. \\
\textbf{Part VI.} 
Practical stability: For $\tau\in[0,T]$ we have:
\begin{align*}
&\ell(x_{\tau|t+h},u_{\tau|t+h})-\ell(x^*_{\tau+h|t},u^*_{\tau+h|t})\\
\stackrel{\eqref{eq:stage_cost_bound}\eqref{eq:prop_robust_1_2_cont}}{\leq}& \alpha_u(e^{-\rho_c \tau} s^*_{h|t}) 
\leq \alpha_u\left(\dfrac{1-e^{-\rho_c h}}{\rho_c}e^{-\rho_c \tau}\overline{w}\right).
\end{align*}
Practical stability follows similar to Theorem~\ref{thm:main}. 
\end{proof}

\begin{remark}
\label{remark:cont}
Compared to Assumption~\ref{ass:contract}, Assumption~\ref{ass:contract_cont} is stronger, as it requires the explicit knowledge of $\kappa$ to be used to compute the closed-loop input~\eqref{eq:close_cont} and we need the special structure $\kappa(x,z,v)=v+\kappa(x)-\kappa(z)$.  
It is possible to relax this assumption to the existence of an exponentially stabilizing open-loop input sequence for any feasible trajectory, similar to Assumption~\ref{ass:contract} (compare also Remark~\ref{rk:stab}). 
In this case, recursive feasibility can be shown similar to Theorem~\ref{thm:main} with the stabilizing input sequence as a feasible candidate solution. 
This, however, requires that the stabilizing input-sequence is in the set of considered input trajectories in the optimization problem~\eqref{eq:MPC_cont_kappa}, which causes difficulties if a simple input parametrization is considered, e.g. piece-wise constant. 
Furthermore, since in this case the feedback is only present at the next sampling instance the initial propagation of the tube changes to 
\begin{align}
&\dot{s}_{\tau|t}=L_c s_{\tau|t}+w_{\tau|t},~\tau\in[0,h],\\
&w_{\tau|t}=\tilde{w}_{\delta,L}(x_{\tau|t},u_{\tau|t},s_{\tau|t}),~\tau\in[0,h],
\end{align}
where the  (possibly positive) constant $L_c$ and function $\tilde{w}_{\delta,L}$ need to satisfy the following condition similar to~\eqref{eq:w_delta_tilde_cont} (in the absence of feedback $\kappa$):
\begin{align}
\label{eq:w_delta_tilde_cont_L}
\dfrac{d}{dt} \sqrt{V_{\delta}(x,z)}\leq w_L\sqrt{V_{\delta}(x,z)}+\tilde{w}_{\delta,L}(z,v,c),\\
\dfrac{d}{dt}x=f(x,v)+d_w,~\dfrac{d}{dt}z=f(z,v).  \nonumber
\end{align}
for all $(x,z,v)\in\mathbb{R}^n\times\mathcal{Z}$. 

In case of piece-wise constant input sequences $u$, the stabilizability condition reduces to Assumption~\ref{ass:contract}, where the discrete-time model is defined as the integration of the continuous-time model and $\rho=e^{-\rho_c h}$. 
Similarly, the discrete-time function $\tilde{w}_{\delta}$ can be defined implicitly based on the continuous-time function $\tilde{w}_{\delta}$ and the integration of $s$ using~\eqref{eq:MPC_cont_s}. 
Although the formulations are equivalent, it may often be more convenient to formulate the function $\tilde{w}_\delta$ in continuous-time and define the discrete-time functions implicitly based on a suitable discretization, compare the numerical example in Section~\ref{sec:num}.  
\end{remark}

\begin{remark}\label{rk:implement}
The dynamic propagation of $s_{\cdot|t}$ with $w_{\cdot|t}$ can pose numerical challenges, since the decision variable $w_{\cdot|t}$ is a continuous-time trajectory.
One solution to this problem is to replace~\eqref{eq:MPC_cont_w} with an equality constraint and thus eliminate $w_{\cdot|t}$ as a decision variable. 
In this case the dynamics of $(x,s)$ in~\eqref{eq:MPC_cont_init}--\eqref{eq:MPC_cont_w} can be jointly discretized. 

However, for many design choices of $\tilde{w}_{\delta}$ it is numerically more efficient to implement~\eqref{eq:MPC_cont_w} using multiple nonlinear inequality constraints, compare e.g.~\eqref{eq:w_LPV_poly}, \eqref{eq:w_tilde_quadrotor}.
In this case it is more natural to also impose a piece-wise constant structure on the decision variable $w_{\cdot|t}$, similar to $v_{\cdot|t}$. 
In this case the considered optimization problem~\eqref{eq:MPC_cont_kappa} has the complexity of a nominal continuous-time MPC with $m+1$ inputs $(v,w)$ and additional nonlinear inequality constraints~\eqref{eq:MPC_cont_w}, similar to the discrete-time case. 
\end{remark}
 
%
\subsection{General nonlinear constraints}
\label{app:nonlin_constraints}
In the following, we show how the proposed approach can be extended to more general nonlinear constraints. 
Consider the following nonlinear constraint set
\begin{align}
\label{eq:gc_new}
\tilde{\mathcal{Z}}_g=\{(x,u)\in\mathcal{Z}|~\tilde{g}_j(x,u)\leq 0,\quad j=1,\dots,q\},
\end{align}
where the functions $\tilde{g}_j$ need not satisfy Assumption~\ref{ass:gen_nonlin_con_Lipschitz}.  
\begin{assumption}
\label{ass:general_constraints}
The functions $\tilde{g}_j$ are continuous, i.e., there exist functions $\alpha_j \in \mathcal{K}_{\infty}$, such that for any $r,~\tilde{r} \in \mathcal{Z}$
\begin{align}
\label{eq:nonlin_con_NotLipschitz}
 \tilde{g}_j(\tilde{r})- \tilde{g}_j(r) \leq \alpha_j(\|r-\tilde{r}\|), \quad j = 1, \dots, q.
\end{align}
Furthermore, there exist constants $\lambda_j > 0$, such that
\begin{align}
\label{eq:lambda}
\alpha_j(\rho c) \leq \rho^{\lambda_j} \alpha_j(c).
\end{align}
\end{assumption}
Condition \eqref{eq:lambda} is satisfied by any polynomial function $\alpha_j$ with positive coefficients.
In the following, we show how the presented approach can be extended to such continuity bounds. 
The consideration of the more general bound~\eqref{eq:nonlin_con_NotLipschitz} instead of the Lipschitz bound~\eqref{eq:nonlin_con_Lipschitz} can also reduce the conservatism. 
\begin{proposition}
Suppose that Assumptions \ref{ass:contract} and \ref{ass:general_constraints} are satisfied, then there exist functions $\tilde{\alpha}_j \in \mathcal{K}_{\infty}$, $j=1, \dots, q$, such that the following inequalities hold for all $(x,z,v)\in \mathbb{R}^n\times\mathcal{Z}$ and all $c\in[0,\sqrt{\delta_{loc}}]$, with $V_\delta(x,z,v)\leq c^2$:
\begin{align}
\label{eq:gc_increase}
\tilde{g}_j(x,\kappa(x,z,v))\leq& \tilde{g}_j(z,v)+\tilde{\alpha}_j (c),\\
\label{eq:lambda_tilde}
\tilde{\alpha}_j(\rho c) \leq &\rho^{\lambda_j} \tilde{\alpha}_j(c).
\end{align}
\end{proposition}
\begin{proof}
The proof is analogous to Proposition \ref{prop:contract}, based on \eqref{eq:bound} and \eqref{eq:k_max}, with 
$\tilde{\alpha}_j(c):=
\alpha_j\left(\sqrt{({1}/{c_{\delta,l}}+\kappa_{max})}c \right)$.
\end{proof}
In order to ensure robust constraint satisfaction for the more general nonlinear constraint set $\tilde{\mathcal{Z}}_g$, the following constraints need to be added to the optimization problem \eqref{eq:MPC_real}:
\begin{subequations}
\label{eq:gc_MPC}
\begin{align}
\label{eq:gc_MPC_a}
&h_{j,k+1|t} = \sum_{i=0}^{k} \tilde{\alpha}_j(\rho^{k-i} w_{i|t}),~h_{j,0|t}=0,\\
\label{eq:gc_MPC_b}
&\tilde{g}_j(x_{k|t}, u_{k|t}) +h_{j,k|t}\leq 0,\\
\label{eq:gc_MPC_d}
&h_{j,N|t} \leq \overline{h}_{j,f},\quad
 k=0,\dots, N-1,\quad j=1,\dots, q.
\end{align}
\end{subequations}
\begin{assumption}
\label{ass:gc_terminal}
Consider the terminal ingredients in Assumption~\ref{ass:term_2}. 
There exist constants $\overline{h}_{j,f}$, $j=1,\dots,q$, such that the following properties hold for any $(x,s)\in \mathcal{X}_f$
\begin{subequations}
\begin{align}
\label{eq:terminal_ing_gc_a}
\tilde{g}_j(x,k_f(x))+\overline{h}_{j,f}\leq& 0,\\
\label{eq:terminal_ing_gc_c}
\tilde{\alpha}_j(\tilde{w}_{\delta}(x,k_f(x),s)) \leq& (1-\rho^{{\lambda}_j})\overline{h}_{j,f} + \rho^{{\lambda}_j} \tilde{\alpha}_j (\rho^{N-1} \overline{w}_{\min}).
\end{align}
\end{subequations}
\end{assumption}
The conditions on $\overline{h}_{j,f}$ can be viewed as an extension of the conditions on the tube size $\overline{s}$ in~\eqref{eq:term_dist_RPI2}, \eqref{eq:term_dist_con2},  \eqref{eq:term_s}. 
\begin{lemma}
\label{lemma:nonlin_con}
Suppose that the conditions in Theorem \ref{thm:main} and Assumptions \ref{ass:general_constraints} and \ref{ass:gc_terminal} are satisfied. Consider the closed-loop system \eqref{eq:close} based on the optimization problem \eqref{eq:MPC_real} with the additional constraints \eqref{eq:gc_MPC}.
 Then this optimization problem is recursively feasible, the constraints~\eqref{eq:gc_new} are satisfied and the origin is practically asymptotically stable for the resulting closed-loop system~\eqref{eq:close}.
\end{lemma}
\begin{proof}
The proof follows the arguments of Theorem \ref{thm:main}.\\
\textbf{Part I.} Candidate solution: analogous to Thm. \ref{thm:main} with $h_{j,\cdot|t+1}$ according to~\eqref{eq:gc_MPC_a}. \\
\textbf{Part II.} Tube dynamics: The following inequality holds for $k=0, \dots, N-1$:
\begin{align}
\label{eq:h_evolution}
&h_{j,k|t+1} \stackrel{\eqref{eq:gc_MPC_a}} = \sum_{i=0}^{k-1} \tilde{\alpha}_j(\rho^{k-i-1} w_{i|t+1})   \\
&\stackrel{\eqref{eq:disturbance_shrinking}}\leq \sum_{i=0}^{k-1} \tilde{\alpha}_j(\rho^{k-i-1} w^*_{i+1|t}) 
\stackrel{\eqref{eq:gc_MPC_a}}=h^*_{j,k+1|t} -\tilde{\alpha}_j(\rho^{k} w^*_{0|t}).\nonumber
\end{align}
\textbf{Part III.} State and input constraint satisfaction \eqref{eq:gc_MPC_b}:\\For $k=0, \dots, N-2$ we have 
\begin{align*}
&\tilde{g}_j(x_{k|t+1}, u_{k|t+1})+h_{j,k|t+1}\\
\stackrel{\eqref{eq:prop_robust_1_2}, \eqref{eq:gc_increase}, \eqref{eq:h_evolution}}\leq&\tilde{g}_j(x^*_{k+1|t},u^*_{k+1|t})+h^*_{j,k+1|t}
\stackrel{\eqref{eq:gc_MPC_b}}{\leq} 0.
\end{align*}
The terminal condition \eqref{eq:gc_MPC_d} ensures constraint satisfaction for $k = N-1$ with
\begin{align*}
&\tilde{g}_j(x_{N-1|t+1}, u_{N-1|t+1})+h_{j,N-1|t+1}\\
 \stackrel{\eqref{eq:prop_robust_1_2}, \eqref{eq:gc_increase}, \eqref{eq:h_evolution}}\leq&\tilde{g}_j(x^*_{N|t}, u^*_{N|t})+h^*_{j,N|t}
 \stackrel{\eqref{eq:terminal_ing_gc_a}, \eqref{eq:gc_MPC_d}} \leq 0.
\end{align*}
\textbf{Part IV.} Terminal tube constraint \eqref{eq:gc_MPC_d}.
Given $(x^*_{N|t},s^*_{N|t})\in\mathcal{X}_f$ and $h^*_{j,N|t}\leq \overline{h}_{j,f}$, we have
\begin{align*}
&h_{j,N|t+1} 
\stackrel{\eqref{eq:gc_MPC_a}}{=} 
 \sum_{i=0}^{N-1}\tilde{\alpha}_j(\rho^{N-i-1}w_{i|t+1})\\
\stackrel{\eqref{eq:disturbance_shrinking}}{\leq} &\sum_{i=0}^{N-1}\tilde{\alpha}_j(\rho^{N-i-1}w^*_{i+1|t})
=  \sum_{i=1}^{N-1} \tilde{\alpha}_j(\rho^{N-i}w^*_{i|t}) + \tilde{\alpha}_j(w^*_{N|t})\\
\stackrel{\eqref{eq:lambda_tilde}}{\leq} & \rho^{{\lambda}_j} \sum_{i=1}^{N-1} \tilde{\alpha}_j(\rho^{N-i-1}w^*_{i|t}) + \tilde{\alpha}_j(w^*_{N|t})\\
\stackrel{\eqref{eq:gc_MPC_a}}{\leq} & \rho^{{\lambda}_j} [ \underbrace{h^*_{j,N|t}}_{ \stackrel{\eqref{eq:gc_MPC_d}}{ \leq }\overline{h}_{j,f}} - \tilde{\alpha}_j(\rho^{N-1}\underbrace{ w^*_{0|t}}_{\stackrel{\eqref{eq:w_min_bound_bar_s}}{\geq} \overline{w}_{\min}} ) ]  + \tilde{\alpha}_j(w^*_{N|t})
\stackrel{\eqref{eq:terminal_ing_gc_c}}{\leq} \overline{h}_{j,f}.
\end{align*}
\end{proof}
\begin{remark}
\label{remark:nonlin_con_Lipschitz}
Lemma~\ref{lemma:nonlin_con} shows that the proposed framework can also ensure robust constraint satisfaction for general nonlinear constraint. 
If $\tilde{\alpha}_j$ is of the form $\tilde{\alpha}_j(r)=\tilde{c}_jr^{\lambda_j}$ with positive constants $\tilde{c}_j,\lambda_j$, the constraints~\eqref{eq:gc_MPC_a} can be replaced by
\begin{align*}
&h_{j,k+1|t}=\rho^{\lambda_j}h_{j,k|t}+\tilde{c}_jw^{\lambda_j}_{i|t},~h_{j,0|t}=0.
\end{align*}
In this case, the computational demand with $q$ general nonlinear constraints and $p$ Lipschitz continuous constraints is equivalent to a nominal MPC scheme with $n+1+q$ states $(x,s,h_j)$ and $m+1$ inputs $(u,w)$. 
Thus, the computational complexity of the robust tube MPC increases with the number of nonlinear constraints $q$ (which are not Lipschitz continuous). 
The proposed framework can also be extended to robust collision avoidance, which requires additional (dual) decision variables, compare~\cite{Soloperto2019Collision}.
\end{remark}

\end{document}